\PassOptionsToPackage{colorlinks=true, linkcolor=blue, citecolor=blue, urlcolor=blue}{hyperref}

\documentclass{article}
\usepackage[english]{babel}

\usepackage[a4paper,top=2cm,bottom=2cm,left=3cm,right=3cm,marginparwidth=1.75cm]{geometry}

\usepackage{graphicx}%
\usepackage{subcaption}
\usepackage{multirow}%
\usepackage{amsmath,amssymb,amsfonts,amsthm,mathrsfs}
\usepackage{siunitx}
\PassOptionsToPackage{hyphens}{url}\usepackage{hyperref}
\usepackage[utf8]{inputenc}
\usepackage[right]{lineno}
\usepackage{csquotes}
\usepackage{booktabs}
\usepackage{longtable}
\usepackage{adjustbox}
\usepackage{array}
\usepackage{url}
\usepackage{titlesec}
\usepackage{authblk}
\usepackage{xcolor} 

\usepackage{textcomp}%
\usepackage{manyfoot}%
\usepackage{booktabs}%
\usepackage{algorithm}%
\usepackage{algorithmicx}%
\usepackage{algpseudocode}%
\usepackage{listings}%
\usepackage{mathtools}
\usepackage{relsize}
\usepackage{multicol}
\usepackage{tabularx}
\usepackage{bm}
 \usepackage{setspace}
 \usepackage{adjustbox}
\usepackage{caption}
\usepackage{booktabs}

\titleformat{\subsection}
  {\mdseries\itshape\large} 
  {\thesubsection}{1em}{} 

\usepackage[english]{babel}
\usepackage[style=authoryear,backend=biber,natbib=true,maxcitenames=2,uniquelist=false]{biblatex}
\DeclareNameAlias{sortname}{family-given}
\DeclareNameAlias{default}{family-given}

\renewbibmacro{in:}{}
\DeclareFieldFormat[article]{title}{\mkbibquote{#1}\addcomma}
\DeclareFieldFormat[book]{title}{\mkbibemph{#1}\addcomma}
\DeclareFieldFormat[bookinbook]{title}{\mkbibemph{#1}\addcomma}
\DeclareFieldFormat[inbook]{title}{\mkbibquote{#1}\addcomma}
\DeclareFieldFormat[incollection]{title}{\mkbibquote{#1}\addcomma}
\DeclareFieldFormat[inproceedings]{title}{\mkbibquote{#1}\addcomma}
\DeclareFieldFormat[manual]{title}{\mkbibemph{#1}\addcomma}
\DeclareFieldFormat[misc]{title}{\mkbibemph{#1}\addcomma}
\DeclareFieldFormat[thesis]{title}{\mkbibemph{#1}\addcomma}
\DeclareFieldFormat[unpublished]{title}{\mkbibquote{#1}\addcomma}
\DeclareFieldFormat[patent]{title}{\mkbibemph{#1}\addcomma}
\DeclareFieldFormat[report]{title}{\mkbibemph{#1}\addcomma}
\DeclareFieldFormat[online]{title}{\mkbibquote{#1}\addcomma}
\DeclareFieldFormat[software]{title}{\mkbibemph{#1}\addcomma}
\DeclareFieldFormat[booklet]{title}{\mkbibemph{#1}\addcomma}
\DeclareFieldFormat[periodical]{title}{\mkbibemph{#1}\addcomma}
\DeclareFieldFormat[standard]{title}{\mkbibemph{#1}\addcomma}

\DeclareFieldFormat[article]{journaltitle}{\iffieldundef{shortjournal}{\mkbibemph{#1}\addcomma}{\mkbibemph{\printfield{shortjournal}}\addcomma}}
\DeclareFieldFormat{volume}{\bibstring{volume}~#1}
\DeclareFieldFormat{number}{\bibstring{number}~#1}

\DefineBibliographyStrings{english}{
  volume = {Vol.},
  number = {No.}
}

\renewbibmacro*{volume+number+eid}{%
  \printfield{volume}%
  \setunit*{\addspace}%
  \printfield{number}%
  \setunit{\addcomma\space}%
  \printfield{eid}}

\renewbibmacro*{journal+issuetitle}{%
  \usebibmacro{journal}%
  \setunit*{\addcomma\space}%
  \usebibmacro{volume+number+eid}%
  \setunit{\addcomma\space}%
  \usebibmacro{issue+date}}

\renewbibmacro*{publisher+location+date}{%
  \printlist{publisher}%
  \iflistundef{location}
    {\setunit*{\addcomma\space}}
    {\setunit*{\addcolon\space}}%
  \printlist{location}%
  \setunit*{\addcomma\space}%
  \usebibmacro{date}}

\DeclareCiteCommand{\cite}[\mkbibparens]  
  {\usebibmacro{prenote}}                  
  {\printtext[bibhyperref]{
     \printnames{labelname}
     \setunit{\nameyeardelim}
     \printfield{labelyear}
     \printfield{extrayear}}}                
  {\multicitedelim}                        
  {\usebibmacro{postnote}}

\renewbibmacro*{cite:labelyear+extrayear}{%
  \iffieldundef{labelyear}
    {}
    {\printtext[bibhyperref]{%
       \printfield{labelyear}%
       \printfield{extrayear}}}}

\renewbibmacro*{cite:labeldate+extradate}{%
  \iffieldundef{labelyear}
    {}
    {\printtext[bibhyperref]{%
       \printfield{labelyear}%
       \printfield{extradate}}}}

\AtEveryBibitem{
  \clearfield{month}
  \clearfield{day}
  \ifentrytype{book}{
    \clearlist{location}
  }{}
}

\DefineBibliographyStrings{english}{
  andothers = {\textit{et al.},}
}

\DeclareFieldFormat[article]{volume}{\bibstring{jourvol}\addnbspace #1}
\DeclareFieldFormat[article]{number}{\bibstring{number}\addnbspace #1}
\DeclareFieldFormat[article]{volume}{Vol. #1}
\DeclareFieldFormat[article]{number}{No. #1}

\DeclareFieldFormat{url}{\bibstring{available at}\addcolon\space\url{#1}}
\DeclareFieldFormat{urldate}{\mkbibparens{accessed \addspace#1}}

\DeclareFieldFormat{urldate}{%
  \mkbibparens{accessed\space%
    \thefield{urlday}\addspace%
    \mkbibmonth{\thefield{urlmonth}}\addspace%
    \thefield{urlyear}}}

\author[1]{ }
\author[2]{ }
\date{}

\title{Efficient Dynamic Rank Aggregation}

\theoremstyle{thmstyleone}%
\newtheorem{theorem}{Theorem}
\theoremstyle{thmstyleone}
\newtheorem{example}{Example}%
\newtheorem{obs}{Observation}
\newtheorem{lemma}{Lemma}
\newtheorem{proposition}{Proposition}

\theoremstyle{thmstylethree}%
\newtheorem{definition}{Definition}%

\newcommand\floor[1]{\lfloor#1\rfloor}

\newenvironment{proofof}[1]{%
  \begin{proof}[Proof of \autoref{#1}]%
}{%
  \end{proof}%
}
\usepackage{hyperref}
\usepackage[nameinlink]{cleveref}

\crefname{section}{Section}{Sections}
\crefname{figure}{Figure}{Figures}
\crefname{table}{Table}{Tables}
\crefname{equation}{Equation}{Equations}
\crefname{example}{Example}{Examples}
\crefname{algorithm}{Algorithm}{Algorithms}
\crefname{theorem}{Theorem}{Theorems}
\crefname{lemma}{Lemma}{Lemmas}
\crefname{remark}{Remark}{Remarks}
\crefname{definition}{Definition}{Definitions}
\crefname{obs}{Observation}{observations }
\title{Efficient Dynamic Rank Aggregation}
\begin{document}
\maketitle




\begin{center}
\vspace*{-2.5cm}

\vspace{1.5cm}

Morteza Alimi $^{a}$(morteza.alimi@uni-a.de), Hourie Mehrabiun $^{b}$ (hourie.mehrabiun@sharif.edu),\\ Alireza Zarei $^{b}$ (zarei@sharif.edu)\\ 

\hspace{10pt}

\begin{flushleft}
\small  
$^a$ Department of Computer Science, University of
Augsburg, Germany. \\
$^b$ Department of Mathematical Sciences, Sharif University of Technology,
Tehran, Iran. \\

\end{flushleft}        
\end{center}

\noindent\textbf{Abstract}

The rank aggregation problem, which has many real-world applications, refers to the process of combining multiple input rankings into a single aggregated ranking. In dynamic settings, where new rankings arrive over time, efficiently updating the aggregated ranking is essential. This paper develops a fast, theoretically and practically efficient dynamic rank aggregation algorithm.

First, we develop the LR-Aggregation algorithm, built on top of the LR-tree data structure, which is itself modeled on the LR-distance, a novel and equivalent take on the classical Spearman's footrule distance. We then analyze the theoretical efficiency of the Pick-A-Perm algorithm, and show how it can be combined with the LR-aggregation algorithm using another data structure that we develop. 

We demonstrate through experimental evaluations that LR-Aggregation produces close to optimal solutions in practice. We show that Pick-A-Perm has a theoretical worst case approximation guarantee of 2. We also show that both the LR-Aggregation and Pick-A-Perm algorithms, as well as the methodology for combining them can be run in $O(n \log n)$ time. 

To the best of our knowledge, this is the first fast, near linear time rank aggregation algorithm in the dynamic setting, having both a theoretical approximation guarantee, and excellent practical performance (much better than the theoretical guarantee). 

\noindent\textbf{keywords} Dynamic Rank Aggregation, Optimal Footrule Aggregation, Approximation Algorithm, Streaming Algorithm


\section{Introduction}
 The rank aggregation problem aims to derive  a consensus ranking by combining multiple input rankings while maximizing consistency according to specific criteria.  Specifically, it combines rankings from various sources into a single aggregated ranking.  For example, ranking   products, hotels, restaurants based on user feedback is a problem in this area.   
  This problem has been widely studied under different names, including the social choice problem~\cite{kemeny2}, consensus ranking~\cite{emond}, preference aggregation~\cite{davenpor}, and Kemeny ranking~\cite{conitzer}.

We first formalize the notation and definition of the rank aggregation problem.

\begin{definition}\label{def:rank}
Let $U = \{e_1, \ldots, e_n\}$ be a universe of $n$ elements. A ranking $\pi$ over $U$ is  a permutation 
 $\pi= \prec e^{\pi}_{1},e^{\pi}_{2},\ldots,e^{\pi}_{n}\succ$, where each  $e^{\pi}_{i}\in U$. 
 Let $U_\pi$ denote the set of elements appearing in $\pi$.  
 The rank ${\pi}(e)$ of an element  $e=e^{\pi}_{i}\in \pi$  is $i$.
 We assume that $e^{\pi}_i$ is ranked higher than $e^{\pi}_j$ in $\pi$
if $i < j$.
  We typically refer to the set $\Pi=[{\pi}_1,\ldots,{\pi}_m]$ of rankings received thus far as the
  \textit{rank aggregation domain}, and denote the \textit{aggregated ranking}
  by $\omega_{\Pi}$.
\end{definition}
Throughout this paper, we  consider only complete rankings,  
meaning rankings $\pi$ for which $U_\pi = U$.  
We hereafter refer to them  simply as rankings.

Various rank aggregation methods have been developed, primarily employing either score-based or distance-based measures.
A well-known  score-based method is the Borda method~\cite{borda}, which assigns scores to elements based on their positions in each ranking.
The aggregated ranking is determined  by total Borda scores.
In distance-based methods, an aggregated ranking is evaluated by its proximity to the input rankings using a distance measure. 
A widely used  distance measure is Spearman's footrule distance~\cite{caid}, defined as follows. 
\begin{equation}\label{eq:spear}
\mathcal{F}(\pi,\pi^\prime)=\sum_{e\in U}|\pi(e)-\pi^\prime(e)|
\end{equation}
where $\pi$ and $\pi^\prime$ are two rankings  over  the same universe $U$ of elements.
   This distance can be computed in linear time~\cite{dwork1}.

 Another distance measure is 
the Kendall-tau distance, which  counts  the number of discordant pairs, where two elements appear in different relative orders across two rankings, defined as:
\begin{align}
\mathcal{K}(\pi,\pi^\prime)=|&\{(e_i,e_j)\quad s.t \quad \pi(e_i) < \pi(e_j) \wedge \pi^\prime(e_i)>\pi^\prime(e_j) \}|
\end{align} 
The Kendall-tau distance between two rankings of size $n$ can be computed in $O(n \log n)$ time~\cite{dwork1}. The 
following relationship has been proved for any pair of  rankings $\pi$ and $\pi^\prime$~\cite{diaconis1}:
  \begin{equation}\label{eq:spearkemeny}
  \mathcal{K}(\pi,\pi^\prime)\leq \mathcal{F}(\pi,\pi^\prime)\leq 2\mathcal{K}(\pi,\pi^\prime)
  \end{equation}
These distances  can  be extended to measure the distance between a ranking  $\pi$ and a list of  rankings $\Pi=[\pi_1\ldots,\pi_m]$  as~follows \footnote{In some references, the \text{normalized} (average) distances are used instead.}:%
\begin{equation}\label{eq:gspear}
  \mathcal{F}(\pi,\Pi)=\sum_{i=1}^{m} \mathcal{F}(\pi,\pi_i)
  \end{equation}
  \begin{equation}\label{eq:gkemeny}
   \mathcal{K}(\pi,\Pi)=\sum_{i=1}^{m} \mathcal{K}(\pi,\pi_i)
 \end{equation}
 The rank aggregation problem has found wide applications across several fields, including social sciences~\cite{arrow,teng,subbian}, bioinformatics~\cite{li,pihur,wang,wang2}, recommendation systems~\cite{oliveira,baltrunas,sohail,balchanowski}, multi-criteria decision making~\cite{mohammadi,dinu}, information retrieval~\cite{farah}, and web search engines~\cite{dwork1,desarkar}.

\subsection{Related work: Static settings}
The well-known problem of rank aggregation has a long history, dating back to the 18th century, when it was studied in the context of social choices~\cite{borda,condorcet}. Several algorithms have been developed to address rank aggregation, using  either score-based or distance-based methods.
The Borda method~\cite{borda} is a widely used score-based rank aggregation technique, where  elements or candidates receive scores based on their ranks in the various input lists. The final aggregated ranking is obtained by sorting the items by their total scores.
In distance-based methods, the rank aggregation problem can be defined as an optimization problem. Given a list of rankings $\Pi$, the rank aggregation problem seeks  an aggregated ranking $\omega$ that minimizes  a distance measure between $\omega$ and  $\Pi$. Well-known distance measures for rank aggregation include   Kendall-tau distance and Spearman's footrule distance. Aggregations obtained by minimizing the footrule distance  as defined  in  \cref{eq:gspear} and Kendall distance  as defined in  \cref{eq:gkemeny} are  called \textit{footrule optimal aggregation} and \textit{Kemeny optimal ranking}, respectively. Bartholdi {\em et al.} proved that finding the Kemeny optimal ranking is NP-hard for complete rankings~\cite{bartholdi}. Dwork {\em et al.}~\cite{dwork1} showed that computing the Kemeny optimal aggregation remains NP-hard even when the number of rankings equals 4. This motivates extensive research into finding an approximate Kemeny optimal aggregation. A well-known algorithm for this is Pick-A-Perm that randomly selects one of the input rankings and is an expected 2-approximation~\cite{ailon}. Dwork {\em et al.} proposed a polynomial time algorithm to compute optimal footrule aggregation~\cite{dwork1}. Based on  \cref{eq:spearkemeny}, optimal footrule aggregation is a $2$-approximation solution for Kemeny optimal aggregation. Coppersmith {\em et al.} proposed the INCR-INDEG algorithm that orders elements based on their weighted in-degrees in a tournament setting~\cite{copper}. This algorithm corresponds closely to Borda’s method, offering a 5-approximation for the Kemeny optimal ranking. There is an experimental comparison of various algorithms for the Kemeny rank aggregation problem by Ali and Meil{\u{a}} in~\cite{ali}.

Several heuristic methods have also emerged that explore rank aggregation from different perspectives.  
Xiao {\em et al.} proposed graph-based methods where rank aggregation is modeled using a graph to represent relationships between rankings~\cite{xiao}. 
Zhang {\em et al.} introduced a graph-based approach for aggregating university rankings, defining a competition graph where nodes represent universities and directed edges indicate outranking relations~\cite{zhang}.
A hierarchical approach has been proposed by Ding {\em et al.} where the aggregation process is broken down into smaller, hierarchical subproblems~\cite{ding}. Schwarz introduced the center ranking aggregation, which aims  to minimize the maximum Kendall-tau distance between the aggregated ranking and all input rankings~\cite{schwarz}. He used the tree search algorithm which is based on a branch-and-bound search to solve this problem. The research by Wei {\em et al.} introduces rank aggregation with proportionate fairness (p-fairness)~\cite{wei}. Their method aims to ensure that the aggregation of items associated with a protected attribute, such as gender, is proportional in every position of the aggregated ranking.  
A popular approach to rank aggregation is based on probabilistic models on permutations, such as the Mallows model~\cite{mallows} and the Luce model~\cite{luce,plackett}. Several research studies in this context include~\cite{qin,lim,zhu,alvin}. Young interpreted Kemeny optimal ranking as a maximum likelihood estimator (MLE), particularly useful in epistemic social choice~\cite{young}. Additionally, a comprehensive survey on rank aggregation, provided in~\cite{wang3}, discusses methodological trends, highlights the strengths and limitations of existing approaches, and identifies key open challenges in the field.

\subsection{Related work: Dynamic settings}
In many real-world scenarios, the aggregated ranking must be updated dynamically as rankings arrive in a stream. The goal is to efficiently maintain the updated aggregation  without recomputing it from scratch. 
Offline algorithms are often impractical due to their high time complexity.  
For example, the polynomial-time algorithm for optimal footrule aggregation~\cite{dwork1}  
relies on computing a minimum perfect matching in a weighted bipartite graph, 
which takes $O(n^{2.5})$ time and is hence
not suitable for the dynamic setting due to its high 
running time. 
Chin~\textit{et al.}~\cite{chin} highlighted the importance of dynamic rank aggregation 
and discussed related challenges but proposed no specific algorithms in this regard. 
In contrast to static settings, few studies address streaming rank aggregation; the studies by Yasutake {\em et al.}~\cite{yasutake} and Irurozki {\em et al.}~\cite{irurozki} (which uses the Mallows model (MM)~\cite{mallows}) are among them. However, the approaches of these papers differ significantly
from ours in terms of input assumptions and optimization goals.  In~\cite{yasutake}, the authors considered a dynamic framework where rankings are received in an infinite stream, a scenario commonly encountered in online learning contexts.
  At each time $t$, a learner 
predicts a permutation $\hat{\sigma}_t\in S_n$, and the adversary provides a true permutation $\sigma_t\in S_n$ to the 
learner. The learner's objective is to minimize  the cumulative loss bound $\sum_{t=1}^{T}\mathcal{K}(\sigma_t,\hat{\sigma}_t)$. 
~\cite{irurozki} addressed  the rank aggregation problem  in  stream learning, assuming rankings are independent and identically distributed (i.i.d) according to the Mallows model (MM)~\cite{mallows}. They considered learning from non-stationary ranking streams,  
where the distribution of preferences changes over time.  
\subsection{Our Contribution} 

We propose a dynamic rank aggregation algorithm with
$O(n \log n)$ update time, presented in three stages.
In~\cref{sec:lr-rank aggregation},  we develop the {\em LR-Aggregation} algorithm, 
which is based on a new definition for Spearman's footrule distance~(\cref{eq:spear}).  As the presented formula for the Spearman’s footrule distance is not well-suited for dynamic settings, we introduce a new distance measure, the \textit{LR-distance}, which we prove to be equivalent to the Spearman’s footrule
distance. This new definition inspired the development of an algorithm for aggregation in dynamic contexts.
The algorithm works by partitioning the domain into two halves,  and then recursing on each half. 
Decisions are made based on the number of occurrences of items in various parts of the input domain, 
which we store in the \textit{LR-tree} data structure. 

The current offline rank aggregation algorithm, which relies on computing a minimum perfect matching, requires $O(n^2.5)$ time to produce the optimal footrule aggregated ranking each time a new ranking is received. In comparison, our approach is significantly more efficient. It updates the aggregated ranking each time a new ranking is received in
just $O(n \log n)$ time and requires only $O(n^2)$ storage. Notably, both the time complexity and storage requirements are independent of the number of rankings, $m$, which may be arbitrarily greater than $n$.

Our experimental evaluations have demonstrated its effectiveness in closely approximating the optimal solution. The practical superiority (small deviation from the optimal solution) and theoretical advantages  (lower time and space complexity)  of our algorithm distinguish it from the best existing  methods, making it suitable for real-world dynamic rank aggregation where timely updates are critical. 

In~\cref{sec:pap}, we show how the Pick-A-Perm algorithm, 
which randomly selects an  input ranking, can be adapted to the dynamic setting. 
Pick-A-Perm has been proven to be an expected 2-approximation for the Kemeny optimal ranking~\cite{ailon}.  
We provide an analysis showing that Pick-A-Perm  
yields an expected 2-approximation under Spearman’s footrule distance. We then show how it can be implemented in the dynamic setting with the same time and space complexity as the LR-Aggregation algorithm, 
by utilizing reservoir sampling.

Our final algorithm, the \textit{Dynamic Rank Aggregation} algorithm, 
runs both the above-mentioned  algorithms and returns the better result.
This algorithm combines the strengths of both previous algorithms 
for efficient computation of optimal footrule aggregation in dynamic setting in $O(n\log n)$ time. 
Having Pick-A-Perm as a subroutine means it is an expected 2-approximation, but 
as our experimental evaluations show, the approximation factor in practice is very close to 1. 

In ~\cref{sec:experiment}, 
we show the practical efficiency of our algorithm by running it on synthetic datasets, 
generated under  uniform, biased, and Mallows model distributions, as well as on a real-world dataset. 
Our experimental setup includes comparing the LR-Aggregation with the optimal footrule aggregation and Pick-A-Perm. 
These experimental results demonstrate that our proposed algorithm 
not only produces close-to-optimal solutions, but also consistently outperforms Pick-A-Perm, 
which justifies the need for the LR-Aggregation algorithm.  In \cref{sec:conclusion}, we discuss the results, present the conclusion, and offer remarks on possible future work.

\section{LR-Aggregation}\label{sec:lr-rank aggregation}
In this section, we develop a fast, space-efficient algorithm for dynamically updating the aggregated ranking.
The main idea is that, for each element $e$, we keep the count of the number of times $e$ appears in total to the left
of a specific position. This information is updated  upon arrival of a new ranking. 
We then use a recursive procedure.  At the top level, we partition the elements based on the number of 
 times they appear in the left part of the domain rankings, and then we recurse on each part. 
We refer to this procedure as \textit{LR-Aggregation}.

  The following definitions formalize this method.
 \begin{definition}
 \label{def:rankinterval}
Given ranking $\pi=\prec e_1,\ldots e_n\succ$, a subranking $\sigma=\prec e_i,e_{i+1},\ldots,e_j\succ$, (where 
$1\leq i \leq j\leq n$) is defined as  a consecutive subsequence of $\pi$. The rank interval  $I(\sigma)$ of a subranking $\sigma$ is 
defined as  the range of  its elements' ranks in $\pi$. In particular,  subrankings of $\pi$  with rank intervals	$[1:
\floor{(n+1)/2}]$ and $[\floor{(n+1)/2}+1:n]$ are  called the left and right subrankings, denoted by $\pi_\ell$ and $\pi_r$,
respectively.

For a subranking $\sigma$ of a ranking $\pi = \prec e_1, \ldots, e_n \succ$ with rank interval $I(\sigma) = [a : b]$, we define its left-extended interval (with respect to $\pi$) as ${I^{\prec_\pi}}(\sigma) = [1 : b]$ and its right-extended interval (with respect to $\pi$) as ${I^{\succ_\pi}} (\sigma)= [a : n]$. The corresponding extended subrankings are denoted by ${\sigma^{\prec_\pi}}$ and ${\sigma^{\succ_\pi}}$, which are the subrankings of $\pi$ restricted to ${I^{\prec_\pi}}(\sigma)$ and ${I^{\succ_\pi}} (\sigma)$, respectively.
\end{definition}
\begin{example}
Consider the ranking $\pi = \prec A, B, C, D, E, F, G, H \succ$.  
 $\sigma = \prec D, E, F \succ$ is a subranking of $\pi$ with 
rank interval $I(\sigma) = [4 : 6]$.  We have:
\begin{align*}
&  I^{\prec_\pi}(\sigma) = [1 : 6],  
&&\sigma^{\prec_\pi} = \prec A, B, C, D, E, F \succ \quad   \\
&I^{\succ_\pi}(\sigma) = [4 : 8], 
&&\sigma^{\succ_\pi} = \prec D, E, F, G, H \succ \quad   
\end{align*}
\end{example}
\begin{definition}
\label{def:presence}
 Given a ranking $\pi$ and an interval $I$, we define the presence of an element $e\in U_\pi$ with respect to $\pi$ and $I$, denoted by   $p(e)^I_\pi$,  as follows:
 \begin{align}
p(e)^I_\pi=&\begin{cases} 
  1,& \text{if}\quad \pi(e)\in I,\\
  0,& \text{otherwise.}
\end{cases}
\end{align}
This definition is generalized to a list of rankings $\Pi$ as follows:
\begin{align}
p(e)^I_{\Pi}=\sum_{\pi\in \Pi}p(e)^I_{\pi}.
\end{align}
If $I$ refers to the entire domain, we simply write $p(e)_{\pi}$ and $p(e)_{\Pi}$.
 \end{definition}
In particular,  the \textit{left presence} and \textit{right presence} of an element $e$ with respect to the ranking $\pi$, denoted by
  $lp(e)_\pi$ and $rp(e)_\pi$\ are defined as:
\begin{equation}
  lp(e)_{\pi}=p(e)_{\pi_\ell}
    \label{rel8}
  \end{equation} 
 \begin{equation}
 rp(e)_{\pi}=p(e)_{\pi_r}
    \label{rel9}
  \end{equation}
The above definitions are also extended for the entire rank aggregation domain $\Pi$, as: 
\begin{equation}
  lp(e)_{\Pi}=\sum_{\pi\in \Pi}lp(e)_{\pi}
    \label{rel7}
  \end{equation} 
 \begin{equation}
 rp(e)_{\Pi}=\sum_{\pi\in \Pi}rp(e)_{\pi}
    \label{rel10}
  \end{equation}
  Now we introduce \textit{Left-Right distance} or \textit{LR-distance} to compute the distance 
between  two rankings. Informally, the LR-distance between two rankings $\pi$ and $\sigma$ is equal to the number of elements in $\pi_\ell$ 
which are not in $\sigma_\ell$, plus the number of elements in $\pi_r$ which are not in $\sigma_r$, plus the distance of $\pi_\ell$ and $\sigma$ and the 
distance of $\pi_r$ and $\sigma$ to be computed recursively. The formal definition follows.
\begin{definition}
Given two complete rankings $\omega$ and $\sigma$ over the same universe $U$ of elements, the LR-distance between a subranking $\pi$ of 
  $\omega$ and ranking $\sigma$ is defined as: 
\begin{equation}\label{def:lr-measure}
\resizebox{\textwidth}{!}{$
\mathcal{LR}(\pi, \sigma) =
\begin{cases}
\sum\limits_{e \in \pi_\ell^{\prec_\omega}} (p(e)_\sigma^{I(\pi_r^{\succ_\omega})}) +
\mathcal{LR}(\pi_\ell, \sigma) +
\sum\limits_{e \in \pi_r^{\succ_\omega}} (p(e)_\sigma^{I(\pi_\ell^{\prec_\omega})}) +
\mathcal{LR}(\pi_r, \sigma), & \text{if $|U_\pi| > 1$},\\
0, & \text{if $|U_\pi| \le 1$}.
\end{cases}
$}
\end{equation}
\end{definition} 
In all recursive calls, $\pi$ is treated as a subranking of the first argument, and the second argument is fixed. Note that if $\pi$ is a complete ranking (i.e., $\pi=\omega$),  the left and right extension notations are considered with respect to the 
complete ranking $\omega$ in all recursive calls of the above definition, and for brevity, we use ${\pi_\ell}^{\prec}$ and 
${\pi_r}^{\succ}$ instead of ${\pi_\ell}^{\prec_\omega}$ and ${\pi_r}^{\succ_\omega}$, respectively.

The above definition is generalized to define the LR-distance between  a ranking $\pi$ and  a list $\Pi=[\pi_1,\pi_2,\ldots,\pi_m]$ of  rankings 
as follows:
\begin{align}\label{def:glr-measure}
\mathcal{LR}(\pi,\Pi)=\sum_{\pi_i\in \Pi}\mathcal{LR}(\pi,\pi_i).
\end{align}
\begin{example}\label{ex:lr}
Consider the rankings $\pi=\prec A, B, C, D, E, F, G, H\succ$ and $\sigma=\prec B, G, A, E, C, F, H, D\succ$. The LR-distance $\mathcal{LR}(\pi,\sigma)$ is computed as follows:
\begin{align*}
\mathcal{LR}(\pi,\sigma)&=2+[\mathcal{LR}(\prec A, B, C, D\succ ,\sigma)]+2+[\mathcal{LR}(\prec E, F, G, H\succ ,\sigma)]\\
&=2+[1+[\mathcal{LR}(\prec A, B \succ ,\sigma)]+1+[\mathcal{LR}(\prec  C, D\succ ,\sigma)]]\\
&+2+[1+(\mathcal{LR}(\prec  E, F\succ ,\sigma)]+1+[\mathcal{LR}(\prec G, H\succ ,\sigma)]]\\
&=2+[1+[1+1]+1+[1+1]]+2+[1+[1+1]+1+[1+1]]=16.
\end{align*}
Here, the first term $2$  represents $\displaystyle {\sum\limits_{e\in {\pi_\ell}^\prec}}(p(e)_\sigma^{I({\pi_r}^{\succ})})$
in the LR-distance relation, and it is computed as follows: 
\begin{align*}
&\pi_\ell = \prec A, B, C, D \succ, 
\quad && {\pi_\ell}^\prec = \prec A, B, C, D \succ, 
\quad && I({\pi_\ell}^\prec) = [1 : 4],\\[4pt]
&\pi_r = \prec E, F, G, H \succ, 
\quad && {\pi_r}^\succ = \prec E, F, G, H \succ, 
\quad && I({\pi_r}^\succ) = [5 : 8],\\[6pt]
& p(A)_\sigma^{I({\pi_r}^\succ)} = 0, 
\quad && p(B)_\sigma^{I({\pi_r}^\succ)} = 0, \\[4pt]
& p(C)_\sigma^{I({\pi_r}^\succ)} = 1, 
\quad && p(D)_\sigma^{I({\pi_r}^\succ)} = 1, \\[6pt]
&\sum_{e \in {\pi_\ell}^\prec} p(e)_\sigma^{I({\pi_r}^\succ)} = 0 + 0 + 1 + 1 = 2.
\end{align*}

\end{example}
The following observations are derived directly from the LR-distance definition.
\begin{obs} For two  rankings $\pi$ and $\sigma$:
 \begin{align}
  \mathcal{LR}(\pi,\sigma)= \mathcal{LR}(\sigma,\pi).
 \end{align}
\end{obs}
\begin{obs}
For a  ranking $\pi$ and  domain $\Pi$ of  rankings we have:
\begin{equation}
\resizebox{0.95\hsize}{!}{
$\displaystyle
\mathcal{LR}(\pi,\Pi)=
\begin{cases}
\sum\limits_{e\in {\pi_\ell}^\prec}(p(e)_\Pi^{I({\pi_r}^{\succ})})+\mathcal{LR}(\pi_\ell,\Pi)+\sum\limits_{e\in {\pi_r}^\succ}(p(e)_{\Pi}^{I({\pi_\ell}^{\prec})})+\mathcal{LR}(\pi_r,\Pi), & \text{if } |U_\pi|>1, \\
0, & \text{if } |U_\pi|\le 1.
\end{cases}
$
}
\end{equation}

\end{obs}

\begin{obs}\label{obs3}
For two complete rankings $\pi$ and $\sigma$ (over the same elements), we have:
\begin{equation}
\begin{aligned}
\mathcal{LR}(\pi,\sigma)&=2\cdot \sum_{e\in {\pi_\ell}^\prec}(p(e)_\sigma^{I({\pi_r}^{\succ})})+\mathcal{LR}(\pi_\ell,\sigma)+\mathcal{LR}(\pi_r,\sigma)&\\
&=2\cdot \sum_{e\in {\pi_r}^\succ}(p(e)_\sigma^{I({\pi_\ell}^{\prec})})+\mathcal{LR}(\pi_\ell,\sigma)+\mathcal{LR}(\pi_r,\sigma).\\
\end{aligned}
\end{equation}
\end{obs}
\begin{proof}
Consider \cref{def:lr-measure}, and let ${\pi_\ell}^\prec=\prec e_1,\ldots, e_k\succ$ and 
${\pi_r}^\succ=\prec e_{k+1},\ldots, e_n\succ$. 
If $\sum\limits_{e\in {\pi_\ell}^\prec}(p(e)_\sigma^{I({\pi_r}^{\succ})})=i$,  only $i$ elements of the subranking ${\pi_\ell}^\prec$  are 
present in the interval $[k+1:n]$ of ranking $\sigma$.  Hence, the remaining $n-k-i$ elements in the interval $[k+1:n]$ of the 
ranking $\sigma$ are present in ${\pi_r}^\succ$. Since $|\{e\in {\pi_r}^\succ\}|=n-k$ and $\sigma$ is a complete ranking,  it follows that
 $(n-k)-(n-k-i)=i$ elements of ${\pi_r}^\succ$ appear in the interval $[1:k]$ of $\sigma$. Therefore, $\sum\limits_{e\in {\pi_r}^\succ}
 (p(e)_\sigma^{I({\pi_\ell}^{\prec})})=i=\sum\limits_{e\in {\pi_\ell}^\prec}(p(e)_\sigma^{I({\pi_r}^{\succ})})$.
\end{proof}
To illustrate the relationship between the   LR distance and Spearman's footrule distance,  consider the  Spearman's footrule distance between the rankings $\pi$ and $\sigma$ from   \cref{ex:lr} shown in \cref{ex:table-example}. In this example, we observe that  $\mathcal{LR}(\pi,\sigma)=\mathcal{F}(\pi,\sigma)$.
 In the following, as the main property of the LR-distance, we prove that the LR-distance between two complete rankings is equal to their 
Spearman's footrule distance.

\begin{table}[h]
\centering
\vspace*{-0.6pt}
\caption{Computing Spearman's footrule distance in \cref{ex:lr}}\label{ex:table-example}
\begin{tabular}{lccccccccc}
\toprule
\textbf{Element}        & \textbf{A} & \textbf{B} &\textbf{ C} & \textbf{D} & \textbf{E} & \textbf{F} & \textbf{G} & \textbf{H} & \textbf{Sum} \\
\midrule
$\bm{\pi(e)}$        & 1 & 2 & 3 & 4 & 5 & 6 & 7 & 8 &  \\
$\bm{\sigma(e)}$        & 3 & 1 & 5 & 8 & 4 & 6 & 2 & 7 &  \\
$\bm{|\pi(e)-\sigma(e)|}$ & 2 & 1 & 2 & 4 & 1 & 0 & 5 & 1 & \textbf{16} \\
\bottomrule
\end{tabular}
\end{table}
\begin{theorem}\label{thm:lrspearman}
For any two  rankings $\pi$ and $\sigma$ over the same universe $U$ (i.e. $U_\pi=U_\sigma$),
\begin{align}
\mathcal{LR}(\pi,\sigma)=\mathcal{F}(\pi,\sigma).
\end{align}
\end{theorem}
First, we observe the following  lemma.
Let $\pi^i$ and ${\bar{\pi}}^i$, for $i\in [1,\ldots ,n-1]$, denote the subrankings of a complete ranking $\pi=\prec e_1,\ldots, e_n\succ$ corresponding to rank intervals  $[1:i]$ and $[i+1:n]$, 
respectively. Assuming $\sigma$  is another complete ranking, we define:
\begin{align*}
d_i(\pi,\sigma):=\sum_{e\in {\pi^i}}(p(e)_\sigma^{I(\bar{\pi}^i)})+\sum_{e\in {\bar{\pi}^i}}(p(e)_{\sigma}^{I(\pi^i)}).
\end{align*}
\begin{lemma}\label{lemma1}
For a ranking $\sigma$ and a subranking  $\omega$ of   ranking $\pi$, with   $I(\omega)=[j:k]$,~we~have: 
\begin{align}\label{lem:lr-relation}
\mathcal{LR}(\omega, \sigma)=\sum_{i=j}^{k-1}d_i(\pi,\sigma).
\end{align}
\end{lemma}
\begin{proof}
The proof is by induction on the recursion depth of the definition of the LR-distance. When the  recursion depth is zero,  $k-j\le 0$, both $\sum_{i=j}^{k-1}d_i(\pi,\sigma)$ and $\mathcal{LR}(\omega, \sigma)$ evaluate to zero, and the statement is trivially true. Now, assume the statement holds for calls at  recursion depth  of $d-1$, in the recursive definition of 
the LR-distance (\cref{def:lr-measure}). Consider a recursive call of LR-distance at depth  $d$ on the ranking 
$\omega=\prec~e_{j},e_{j+1},\ldots,e_{k}\succ$.
According to the recursive definition:
\begin{align*}
\mathcal{LR}(\omega,\sigma)=\sum\limits_{e\in {\omega_\ell}^{\prec_\pi}}(p(e)_\sigma^{I({\omega_r}^{\succ_\pi})})+\mathcal{LR}
(\omega_\ell,\sigma)+~\sum\limits_{e\in {\omega_r}^{\succ_\pi}}(p(e)_{\sigma}^{I({\omega_\ell}^{\prec_\pi})})+\mathcal{LR}(\omega_r,\sigma).
\end{align*}
Because $\mathcal{LR}(\omega_\ell,\sigma)$ and $\mathcal{LR}(\omega_r,\sigma)$ are recursive calls with depth  $d-1$, in which $\omega_\ell=\prec 
e_{j},e_{j+1},\ldots,\\e_{\floor{(j+k)/2}}\succ$ and 
$\omega_r=\prec e_{\floor{(j+k)/2}+1},\ldots,e_{k}\succ$, by induction we have:

\noindent\begin{tabularx}{\textwidth}{@{}XX@{}}
\begin{equation}
  \mathcal{LR}(\omega_\ell,\sigma)=\sum_{i=j}^{\floor{(j+k)/2}-1}d_i(\pi,\sigma) 
    \label{eq:main1}
  \end{equation} &
 \begin{equation}
 \mathcal{LR}(\omega_r,\sigma)=\sum_{i=\floor{(j+k)/2}+1}^{k-1}d_i(\pi,\sigma)
    \label{eq:main2}
  \end{equation}
  \end{tabularx}

Let $m:=\floor{(j+k)/2}$.
Note that ${\omega_\ell}^{\prec_\pi}$ and $ {\omega_r}^{\succ_\pi}$ are equal to $\pi^{m}$ and $\bar{\pi}^{m}$, respectively. Hence, 
$I({\omega_\ell}^{\prec_\pi})$ and  $I({\omega_r}^{\succ_\pi})$ are equal to $I(\pi^m)$ and $I(\bar{\pi}^{m})$, respectively. Therefore, we have:
\begin{small}
\begin{align*}
\sum\limits_{e\in {\omega_\ell}^{\prec_\pi}}(p(e)_\sigma^{I({\omega_r}^{\succ_\pi})})&+\sum\limits_{e\in {\omega_r}^{\succ_\pi}}
(p(e)_{\sigma}^{I({\omega_\ell}^{\prec_\pi})})=\\
\sum_{e\in {\pi^m}}(p(e)_\sigma^{I(\overline{\pi}^{m})})\quad & +\sum_{e\in {\overline{\pi}^m}}(p(e)_{\sigma}^{I(\pi^m)})\quad 
=d_m(\pi,\sigma)=d_{\floor{(j+k)/2}}(\pi,\sigma).
\end{align*}
\end{small}
Together with  \cref{eq:main1} and \cref{eq:main2}, we conclude that $\mathcal{LR}(\omega,\sigma)=\sum_{i=j}^{k-1}d_i(\pi,\sigma) $ which completes the proof of the lemma.    
\end{proof}
\begin{proofof}{thm:lrspearman}
As a consequence of the above lemma, we have
\begin{align}\label{lem:lr-relationmain}
\mathcal{LR}(\pi,\sigma)=\sum_{i=1}^{n-1}d_i(\pi,\sigma),
\end{align}
where $n$ is the size of the rankings. Consider an arbitrary element $e\in \pi$. Recall that $\pi(e)$ and $\sigma(e)$ denote  the positions 
of $e$ in the rankings $\pi$ and $\sigma$, respectively. Each element $e$ contributes one to the terms $d_i(\pi,\sigma)$ of   \cref{lem:lr-relationmain} with  $\pi(e) \leq i<\sigma(e)$ or $t(e) \leq i<\pi(e)$. The number of such positions $i$ is equal to $|\pi(e)-\sigma(e)|$. 
Therefore, the number of such terms and hence the contribution of each element $e$ to the LR-distance is equal to $|\pi(e)-\sigma(e)|$. Summing this over all elements $e$, we obtain: 
\begin{align*}
\mathcal{LR}(\pi,\sigma)=\sum_{e}|\pi(e)-\sigma(e)|=\mathcal{F}(\pi,\sigma).
\end{align*}
\end{proofof}
~\cref{thm:lrspearman} implies minimizing the LR-distance is equivalent to minimizing Spearman’s footrule distance. 
LR-distance definition enables the design of an efficient algorithm  
for Spearman's footrule-based rank aggregation in dynamic settings. Since intuition behind the LR-distance forms the basis of the proposed LR-Aggregation algorithm,  
we refer to the problem of finding a  ranking $\pi$  
that minimizes $\mathcal{LR}(\pi, \Pi)$ as the  
\textit{LR-Aggregation problem} rather than the  
Spearman aggregation problem.

\subsection{The Data Structure for LR-Aggregation}\label{subsec:datastructure}
We introduce a data structure for efficient LR-Aggregation.
Consider a  domain $\Pi$ of  rankings over elements $e_1,\ldots,e_n$.
We  assume   $n=2^k$ for some integer $k\geq 1$, as padding with dummy elements ensures this condition without affecting  the  complexity.

Each new ranking updates the element presence values  with respect  to the  domain, impacting the aggregated ranking.
To efficiently access and update  these values, we employ a balanced binary search tree $T(e)$, called an \textit{LR-tree},  for each element $e$.
The  LR-tree  of each element  stores its  presence values in different intervals of the rankings.

Access to LR-trees is facilitated by an array $\mathbf{R}_{LR}$, called  the \textit{LRroots array}, which stores pointers to their roots.
Each node $u$ in $T(e)$ corresponds to an interval $I$ and stores $p(e)^{I}_\Pi$.
The root node corresponds to the  full interval $I=[1:n]$, and  stores the presence of $e$ with respect to $\Pi$ and $I$.
The  tree is built recursively  by dividing each node’s interval into left and right parts,  and storing $p(e)^{I_{left}}_\Pi$ and $p(e)^{I_{right}}_\Pi$ at the left   and right child  nodes, respectively. 

For non-leaf nodes, we define a \textit{score} attribute, used in the LR-Aggregation process.
Given a non-leaf node $u$ with $I(u) = [a: b]$, we define $I^\prec(u) = [1: b]$, analogous to \cref{def:rankinterval}, 
and compute its score as:
\begin{align}
\text{score}(u) = p(e)_{\Pi}^{I^{\prec}(lc(u))}.
\end{align}

Thus, each node $u$ in $T(e)$ has the following attributes: 
left child $lc(u)$, right child $rc(u)$, parent $par(u)$, interval $I(u)$, 
presence value $p(u) = p(e)^{I(u)}_\Pi$, and  $\text{score}(u)$.

The following observation directly follows from the score definition.
\begin{obs}\label{obs:score} For a node $u$ in an LR-tree, 

$\bullet$ If $u$ is the root,  then $score(u)=p\big(lc(u)\big)$.

$\bullet$ If $u$ is a left child, then $score(u)=score\big(par(u)\big)-p\big(rc(u)\big)$.

$\bullet$ If $u$ is a right child, then $score(u)=score\big(par(u)\big)+p\big(lc(u)\big)$.
\end{obs} 
\cref{alg:score} formally describes the score computation. To demonstrate this process, consider a portion of an LR-tree as shown in \cref{fig:lr-tree}.
\begin{align*}
&score(v)=p(v_1)=15.\\
&score(v_1)=p(v_3)=4.\\
&score(v_4)=p(v_3)+p(v_5)=10.
\end{align*}
\begin{figure}[H]
\centering
\includegraphics[scale=0.4]{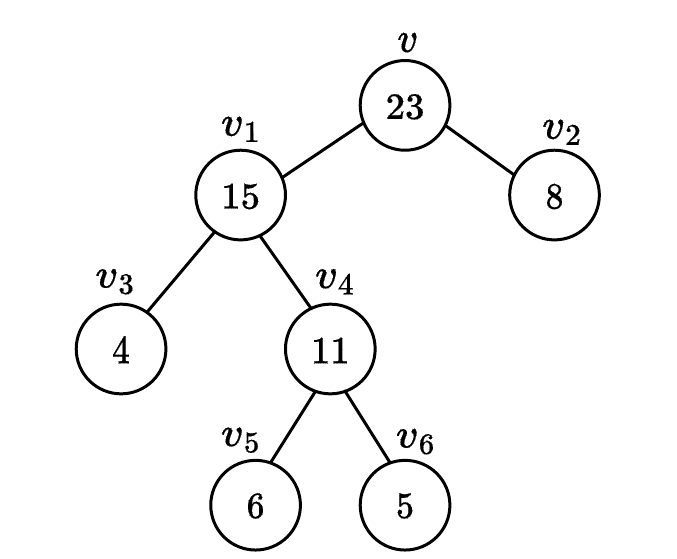}\caption{Score computation of a non-leaf  node in LR-tree $T(e)$}\label{fig:lr-tree}
\end{figure}
\begin{algorithm}[H]
\caption{\textsc{Score}($T$, $u$)} 
\label{alg:score}
\begin{algorithmic}[1]
\Require LR-tree $T$ and a node $u$
\Ensure $\text{score}(u)$ in $T$
\If{$par(u) = \text{null}$}
    \State \Return $p(lc(u))$
\Else
    \If{$u = lc(par(u))$}
        \State \Return \textsc{Score}$\left(T, par(u)\right)$ $-\, p(rc(u))$  
    \Else
        \State \Return \textsc{Score}$\left(T, par(u)\right)$ $+\, p(lc(u))$
    \EndIf
\EndIf
\end{algorithmic}
\end{algorithm}
The initial LR-trees are constructed recursively from the first ranking $\pi = \pi_1$,  
where each root stores $p(e)_\pi$, and subtrees are built from subrankings $\pi_\ell$ and $\pi_r$.  
\cref{fig:2}  illustrates initial LR-trees for elements $C$ and $G$ based on the initial ranking $\pi=\prec A, B, C, D, E, F, G, H \succ$   and \cref{alg:buildtree}  outlines the initial tree construction procedure.
\begin{figure}[H]
\centering
\begin{subfigure}{0.45\textwidth}
  \centering
  \includegraphics[width=0.8\linewidth]{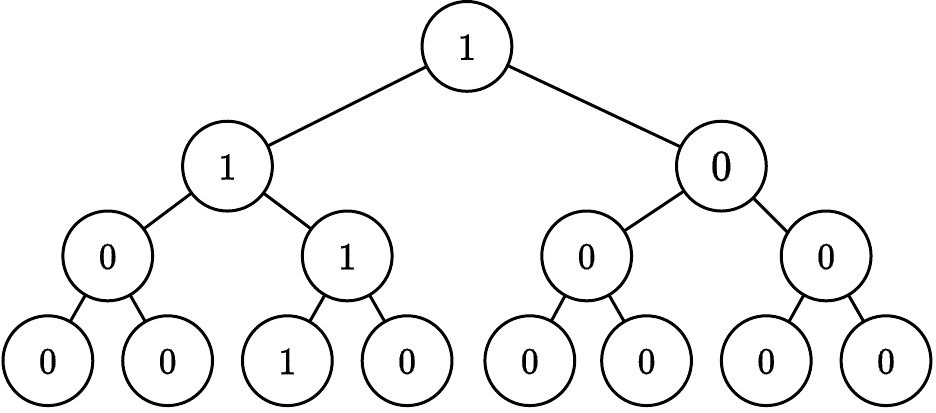}
  \caption*{$T(C)$}
\end{subfigure}%
\hspace{-0.8cm}
\begin{subfigure}{0.45\textwidth}
  \centering
  \includegraphics[width=0.8\linewidth]{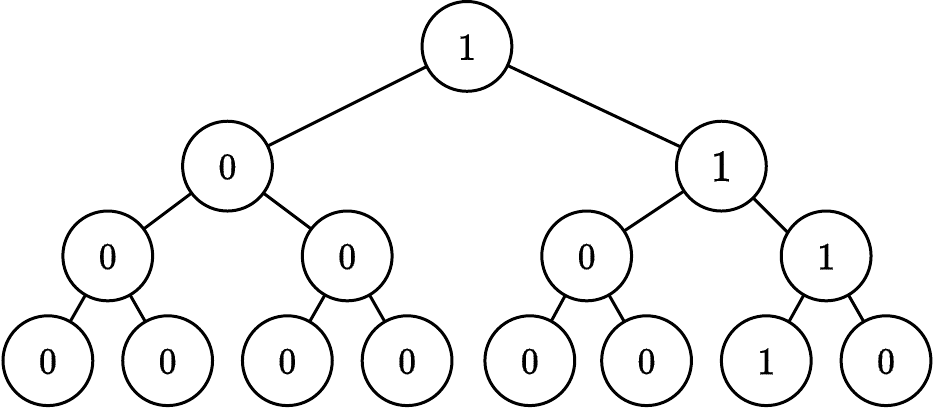}
  \caption*{$T(G)$}
\end{subfigure}
\caption{The initial LR-trees of elements $C$ and $G$ according to initial  ranking $\pi=\prec A, B, C, D, E, F, G, H\succ$}
\label{fig:2}
\end{figure}

\floatname{algorithm}{Algorithm}
\begin{algorithm}[H]
\caption{\textsc{Build LR-Tree}$(e, \pi)$}
\label{alg:buildtree}
\begin{algorithmic}[1]
\Require An element $e$ and a ranking $\pi$
\Ensure The root of the initial LR-tree $T(e)$
\If{$|U_\pi| = 0$}
    \State \Return $null$
\Else
    \State Create node $v$
    \State $p(v) \gets p(e)_\pi$
    \State $I(v) \gets I(\pi)$
    \State $lc(v) \gets$ \textsc{Build LR-Tree}$(e, \pi_\ell)$
    \State $rc(v) \gets$ \textsc{Build LR-Tree}$(e, \pi_r)$
    \State \Return $v$
\EndIf
\end{algorithmic}
\end{algorithm}
To update the LR-tree $T(e_i)$ upon receiving a new ranking, we increment the presence values of nodes whose intervals contain $\pi(e_i)$.
Starting from the root, we traverse down to a leaf node by selecting the child whose interval includes $\pi(e_i)$, incrementing presence value  by one along the path.
~\cref{alg:updatetree} formally describes this update procedure.
To update all LR-trees after receiving a ranking $\pi=\prec e_1,\ldots, e_n\succ$, we apply this procedure to each $T(e_i)$, for all $1\leq i\leq n$.

\floatname{algorithm}{Algorithm}
\begin{algorithm}[H]
\caption{\textsc{Update LR-Tree}$(v, \pi(e), n)$}
\label{alg:updatetree}
\begin{algorithmic}[1]
\Require Root $v$ of LR-tree $T(e)$, the rank $\pi(e)$ of element $e$, and the size of ranking $\pi$
\Ensure Updated LR-tree $T(e)$
\State $p(v) \gets p(v) + 1$
\State $j \gets 1$, $mid \gets n / 2$
\While{$v$ is not a leaf}
    \State $j \gets j + 1$
    \If{$\pi(e) \leq mid$}
        \State $v \gets lc(v)$
        \State $mid \gets mid - n / 2^j$
    \Else
        \State $v \gets rc(v)$
        \State $mid \gets mid + n / 2^j$
    \EndIf
    \State $p(v) \gets p(v) + 1$
\EndWhile
\end{algorithmic}
\end{algorithm}

\begin{theorem}\label{thm:updatelrtree}
 The set of LR-trees for a rank aggregation domain $\Pi=[\pi_1,\dots,\pi_m]$ over $n$ elements requires $O(n^2)$ space and updates in $O(n \log n)$ time per new ranking $\pi=\prec e_1,\dots,e_n\succ$. \footnote{This is based on the assumption that counters take $O(1)$ space irrespective of how many times they are
incremented (i.e.\ the \textit{Real RAM} Model of computation). If the number of bits in a counter is also taken 
into account (i.e.\ the \textit{Word RAM} model of computation), then a multiplicative factor of $\log m$ also 
has to be included in space- and time-complexities.}
\end{theorem}
\begin{proof}
Note that the leaves  in each LR-tree correspond to the intervals in the form of $[i]$, representing a position in the ranking. Consequently, each LR-tree has 
$n$ leaves. Since each LR-tree is a complete binary tree  and each node (both leaves and internal) uses 
$O(1)$ storage,  each LR-tree requires  $O(n)$ storage. Therefore, the total amount of required storage for storing all LR-trees
corresponding to $n$ elements is $O(n^2)$.
 
 The update process involves traversing each LR-tree, where the path length is $O(\log n)$. Since updating each node requires $O(1)$ time, updating each tree takes  $O(\log n)$ time. Using the pointers to the roots of LR-trees stored in  array $\mathbf{R}_{LR}$, each tree can be accessed  in $O(1)$ time. Consequently, the total time required for updating the entire  set of  LR-trees is $O(n\log n)$%
\end{proof}
\subsection{LR-Aggregation Algorithm}\label{subsec-dyrankagg}
Given a rank  aggregation domain  $\Pi$, the optimal LR-Aggregation that we are looking for is a ranking $\pi$ that 
 minimizes $\mathcal{LR}(\pi,\Pi)$, i.e.,
\begin{equation*}
\sum\limits_{e\in {{{\pi_\ell}^{\prec}}}}(p(e)_\Pi^{I({\pi_r}^{\succ})})+\mathcal{LR}(\pi_\ell,\Pi)+\sum\limits_{e\in 
{{{\pi_r}^{\succ}}}}(p(e)_{\Pi}^{I({\pi_\ell}^{\prec})})+\mathcal{LR}(\pi_r,\Pi).
\end{equation*}
Partitioning elements by their LR-tree root scores (${lp(e)}_\Pi$) around the median, placing higher-score elements in the left partition and lower-score elements in the right, yields  a ranking $\pi$ that minimizes:
\begin{equation*}
\sum\limits_{e\in {{{\pi_\ell}^{\prec}}}}(p(e)_\Pi^{I({\pi_r}^{\succ})})+\sum\limits_{e\in {{\pi_r}^{\succ}}}
(p(e)_{\Pi}^{I({\pi_\ell}^{\prec})}).
\end{equation*}
However, the definition of LR-distance is recursive. The  aforementioned value equals  the  contribution of $\mathcal{LR}(\pi,\Pi)$
in the initial recursive call of the distance measure. To apply recursive definition  in subsequent levels, we repeat the  previous partitioning process independently on the left and right halves of the obtained ranking, as detailed below. 
Note that the ranking obtained by merging $\pi_\ell$ and $\pi_r$ may no longer minimize the previous value, so repeating this procedure may not always yield the optimal LR-Aggregation. However, as we will see, experimental comparisons demonstrate that the LR distance between the output of this recursive algorithm and the domain is very close to the optimal value.

We now describe the LR-Aggregation algorithm’s steps in detail.
Let  $\omega_m$ be the current aggregated ranking of $\Pi_m=[\pi_1,\ldots,\pi_{m}]$.  The elements are stored in the  aggregation array,
$\mathbf{A}_{agg}$, in the order of their presence in $\omega_m$. Upon  receiving $\pi_{m+1}=\prec e_1,\ldots e_n\succ$, $\mathbf{A}_{agg}$ is updated to obtain the new aggregated ranking $\omega_{m+1}$. This update follows two main steps:

First, the LR-trees of elements in $\pi_{m+1}$ are updated as described in \cref{subsec:datastructure}. 
Then, $\mathbf{A}_{agg}$ is reordered based on the updated LR-trees, producing  $\omega_{m+1}$.

A recursive procedure reorders $\mathbf{A}_{agg}$ using scores computed at each recursion level. 
These scores are derived from LR-tree traversals, where at each recursion level, an element's score equals the score of its current node on the traversal path.
Traversal starts at the roots of LR-trees, computing $score(root(T(e_i)))$ for $1\leq i\leq n$. 
Elements in $\mathbf{A}_{agg}$ are then partitioned around the element with the median score. In cases where two or more elements  receive the same score, we resolve this by using the order of the
elements in the most recent aggregated ranking as the baseline, and maintain the elements in that order.
The recursion then proceeds on the left and right subarrays of $\mathbf{A}_{agg}$, with traversal paths continuing to the left and right subtrees of the LR-trees corresponding to the elements  in $\mathbf{A}_{agg}[1, \ldots, \lfloor n/2 \rfloor]$ and $\mathbf{A}_{agg}[\lfloor n/2 \rfloor + 1, \ldots, n]$, respectively.
As  traversal paths reach the second level of all trees, the recursive procedure repeats. During each recursive call, the score of every element is updated based on its current node, and the corresponding subarrays $\mathbf{A}_{agg}[1, \ldots, \lfloor n/2 \rfloor]$ and $\mathbf{A}_{agg}[\lfloor n/2 \rfloor + 1, \ldots, n]$ are reordered accordingly. 

This process continues recursively until the traversal paths reach the leaf nodes of  LR-trees. At this point, $\mathbf{A}_{agg}$ contains the final aggregated ranking $\omega_{m+1}$. 
 Each entry in $\mathbf{A}_{agg}$ stores a pointer $cr$ to its current  node in the LR-tree, initialized to the root of the corresponding LR-tree, enabling efficient traversal updates. 
It also maintains the $element$ and  $score$ attributes.

Pseudo-code for the  LR-Aggregation algorithm and the   aggregation array update procedure is provided in \cref{alg:dyrankagg,alg:uprankaggarray}.
\begin{algorithm}[H]
\caption{\textsc{LR-Aggregation}$(\mathbf{A}_{agg}, \pi, \mathbf{R}_{LR})$}
\begin{spacing}{0.95}
\label{alg:dyrankagg}
\begin{algorithmic}[1]
\Require Aggregation array $\mathbf{A}_{agg}$, input ranking $\pi$, LRroots array $\mathbf{R}_{LR}$
\Ensure Updated  aggregation array $\mathbf{A}_{agg}$
\State $n \gets$ size of ranking $\pi$
\For{$j \gets 1$ to $n$}
    \State $e \gets$ $j$-th element of $\pi$
    \State $v \gets \mathbf{R}_{LR}[e]$
    \State \textsc{Update LR-Tree}$(v, j, n)$
\EndFor
\For{$j \gets 1$ to $n$}
    \State $e \gets \mathbf{A}_{agg}[j].element$
    \State $\mathbf{A}_{agg}[j].cr \gets \mathbf{R}_{LR}[e]$
    \State $\mathbf{A}_{agg}[j].score \gets p(lc(\mathbf{R}_{LR}[e]))$
\EndFor
\State \textsc{Update  Aggregation Array}$(\mathbf{A}_{agg}, 1, n)$
\State\Return $\mathbf{A}_{agg}$
\end{algorithmic}
\end{spacing}
\end{algorithm}

\floatname{algorithm}{Algorithm}
\begin{algorithm}[H]
\caption{\textsc{Update  Aggregation Array}$(\mathbf{A}_{agg}, start, end)$}
\label{alg:uprankaggarray}
\begin{algorithmic}[1]
\Require Aggregation array $\mathbf{A}_{agg}$ and indices $start$ and $end$
\Ensure Updated  aggregation array
\State $mid \gets \floor{(start + end)/2}$
\State Partition $\mathbf{A}_{agg}$ around the median-score element
\Comment{The top-score elements are placed in $[start : mid]$}
\If{$start + 1 < end$}
    \For{$j \gets start$ to $mid$}
        \State $s \gets \mathbf{A}_{agg}[j].score$
        \State $\mathbf{A}_{agg}[j].cr \gets lc(\mathbf{A}_{agg}[j].cr)$
        \State $\mathbf{A}_{agg}[j].score \gets s - p(rc(\mathbf{A}_{agg}[j].cr))$
    \EndFor
    \For{$j \gets mid + 1$ to $end$}
        \State $s \gets \mathbf{A}_{agg}[j].score$
        \State $\mathbf{A}_{agg}[j].cr \gets rc(\mathbf{A}_{agg}[j].cr)$
        \State $\mathbf{A}_{agg}[j].score \gets s + p(lc(\mathbf{A}_{agg}[j].cr))$
    \EndFor  
    \State \textsc{Update    Aggregation Array}$(\mathbf{A}_{agg}, start, mid)$
    \State \textsc{Update    Aggregation Array}$(\mathbf{A}_{agg}, mid+1, end)$
\EndIf   
\end{algorithmic}
\end{algorithm}

\begin{theorem}\label{thm:dyrankagg}
Let $\omega_m$ be the aggregation of rankings in the rank aggregation domain $\Pi_m=[\pi_1,\ldots, \pi_m]$. The LR-Aggregation algorithm provides the aggregation $\omega_{m+1}$ in $O(n\log n)$ time after receiving $\pi_{m+1}$. 
\end{theorem}

\begin{proof}
    Updating the current  aggregated ranking $\omega_m$ after receiving the  new ranking $\pi_{m+1}$ involves two main steps. The first step is updating the 
set of LR-trees, which takes  $O(n\log n)$ time by  
\cref{thm:updatelrtree}. The second step is updating the current  aggregation array $\mathbf{A}_{agg}$. Before 
updating $\mathbf{A}_{agg}$, the roots of LR-trees and their scores are stored in   $\mathbf{A}_{agg}$.     Note that using the LRroots 
array, we can access the root of each LR-tree in $O(1)$, and their scores are also computed in $O(1)$ time. Hence,  before updating
the  aggregation array, we spend  $O(n\log n+n)=O(n\log n)$ time.

 We now bound the time for the second step.
 Finding the median and partitioning the array of size $O(n)$ takes $O(n)$ time. From the \cref{obs:score}, the scores of the current nodes can be
 computed using the  score of their parents in  $O(1)$ time. So, the total time taken by lines 4-11 in  \cref{alg:uprankaggarray} is  $O(n)$ resulting 
 in  $O(n)$ time before the two recursive calls of the algorithm. This leads to the following recurrence for the total running time of the second 
 step
\begin{equation}
T(n)=O(n)+2T(n/2)
\end{equation}
which  solves to $O(n\log n)$.
Therefore, the total time required to update the existing aggregated ranking is $O(n\log n)$.
\end{proof}
In ~\cref{sec:experiment} we will demonstrate the practical efficiency of the LR-Aggregation algorithm. 
First, however, we develop our \textit{Dynamic Rank Aggregation} algorithm, by showing how LR-Aggregation can be combined with Pick-A-Perm. 
\section{Dynamic Rank Aggregation Algorithm}\label{sec:pap}
In this section, we develop the \textit{Dynamic Rank Aggregation}  algorithm, 
which builds on the LR-Aggregation algorithm by utilizing the well-known Pick-A-Perm algorithm, 
in order to get a theoretical approximation guarantee without compromising on output quality or running time. 

Pick-A-Perm randomly selects one of the input rankings and is an expected
2-approximation for the Kemeny optimal ranking~\cite{ailon}.
We demonstrate that it achieves an expected 2-approximation for footrule optimal aggregation.
Combining the two algorithms 
requires showing: 1. Pick-A-Perm can be made to run in $O(n\log n)$ time in the dynamic setting, 2. The footrule cost of a candidate ranking, with respect to a domain, can be computed in $O(n\log n)$ time, 
so that the best between the outputs of the two algorithms can be chosen
without compromising on the asymptotic efficiency of the running time.

First, we observe that Pick-A-Perm gives an [expected] 2-approximation  for  Spearman's footrule distance. 
We will use the following result from~\cite{dwork1}.
\begin{proposition}\cite{dwork1}\label{thm:median}
Given complete rankings $\pi_1,\ldots, \pi_m$, if the median positions of the elements  form a permutation, this 
permutation is a footrule optimal aggregation.
\end{proposition}
\begin{proposition}
[Pick-A-Perm is 2-approximation for footrule]\label{thm:pratio}
 Let  $\Pi=[\pi_1,\pi_2,\ldots, \pi_m]$ be a rank aggregation domain including  rankings over an element set $U$.  If 
 $\pi^*$ is the optimal footrule aggregation of $\Pi$ and $\pi$ is a random ranking uniformly selected from $\Pi$, then
 \begin{equation}
E[\mathcal{F}(\pi,\Pi)]\leq 2 \mathcal{F}(\pi^*,\Pi).
 \end{equation}
\end{proposition}
\begin{proof}\label{thm:proofratio}
Let $\pi$ be selected from $\Pi$ uniformly at random.
Consider an arbitrary element  $e$ and let its positions in the rankings  $\pi_1,\pi_2,\ldots \pi_m$ be  $a_1,a_2,\ldots,a_m$. 
Without loss of generality, assume  $a_1\leq a_2\leq \ldots\leq a_m$. 
By \cref{thm:median}, the minimum cost $C^*_e$ incurred by an element $e$ 
to the domain occurs when $e$ is placed at $a_t=a_{\lfloor m/2\rfloor}$. Hence,
\begin{small}
  \begin{align*}
C^*_e=\sum_{i=1}^{m}|a_t-a_i|&=\sum_{i=1}^{t-1}(a_t-a_i)+\sum_{i=t+1}^{m}(a_i-a_t)\\
&=\sum_{i=1}^{t-1}i*(a_{i+1}-a_i)+\sum_{i=t}^{m-1}(m-i)*(a_{i+1}-a_i)
=\sum_{i=1}^{m-1}\gamma_i^*(a_{i+1}-a_i),
\end{align*}  
\end{small}
where  $\gamma^*_i=i$ for $i<t$  and $\gamma^*_i=(m-i)$ for $i> t$.
Let $C^i_e$ be the total distance of $e$ to the domain  
when $e$ is at position $a_i$, 
and let $C_e$ be $e$'s cost in $\pi$.  Then 
\begin{align}
E[C_e]&=\sum_{i=1}^{m}(C^i_e) Pr(\text{$e$ be at position $a_i$})=\sum_{i=1}^{m}\dfrac{1}{m}(C^i_e)\\
&=\sum_{i=1}^{m}\frac{1}{m}[\sum_{j=1}^{i-1}(a_i-a_j)+\sum_{j=i+1}^{m}(a_j-a_i)]\label{re2}\\
&=\sum_{i=1}^{m-1}\frac{1}{m}\gamma_i(a_{i+1}-a_i)
\end{align}
To derive a relation between $C^*_e$ and $E[C_e]$, we show that $\gamma_i\leq 2mi$ for $i<t$ and $\gamma_i\leq 2m(m-i)$ for $i> t$. 
Consider a term of $(a_k-a_{k^\prime})$ in \cref{re2}. 
For each pair $(k, k')$, where $k > k'$, the term $(a_k - a_{k'})$ appears twice in  \cref{re2}. One occurrence corresponds to  
$C^k_e$ and the other corresponds to $C^{k^\prime}_e$.
The expression $(a_k - a_{k'})$ can be alternatively written as 
$(a_k-a_{k^\prime})=\sum_{i=k^\prime}^{k-1}(a_{i+1}-a_{i})$.
For every  $i$, the smaller terms in the form of $(a_{i+1}-a_{i})$  contribute to the decomposition of $(a_k-a_{k^\prime})$ 
for pairs $(k,k^\prime)$ where $k>i$ and $k^\prime \leq i$. The number of such pairs equals  $i\cdot(m-i)$. Therefore, for each 
$i,$ the number of terms in the form $(a_{i+1}-a_{i})$ in \cref{re2}, or equivalently $\gamma_i$, equals $2i(m-i)$. Hence,
\begin{align*}
E[C_e]&=\sum_{i=1}^{t-1}\frac{1}{m}\gamma_i(a_{i+1}-a_i)+\sum_{i=t}^{m-1}\frac{1}{m}\gamma_i(a_{i+1}-a_i)\\
&\leq \sum_{i=1}^{t-1}2i(a_{i+1}-a_i)+\sum_{i=t}^{m-1}2(m-i)(a_{i+1}-a_i)=2C^*_e
\end{align*}
So, for a uniformly selected random ranking $\pi$ we have,
\begin{equation}
E[F(\pi,\Pi)]=E[\sum_{e\in U}C_e]=\sum_{e\in U}E[C_e]\leq \sum_{e\in U}2C^*_e=2\sum_{e\in U}C^*_e=2\mathcal{F}
(\pi^*,\Pi)
\end{equation}
\end{proof}
\begin{theorem}
The Dynamic Rank Aggregation algorithm (\cref{alg:best_of}), 
which runs both Pick-A-Perm and  LR-Aggregation and selects  the best solution,
is an expected 2-approximation  for optimal footrule aggregation.
\end{theorem}
 We now show the above-mentioned algorithm takes $O(n \log n)$ time in the dynamic setting. 
The first challenge can be dealt with using reservoir sampling. 
Suppose after $i-1$ iterations, we have a random sample $\tau$ of the first $i-1$ input rankings. 
After receiving $\pi_{i}$, we keep $\tau$ with probability $\frac{i-1}{i}$, and update $\tau$ to be $\pi_i$ with probability $\frac{1}{i}$. 
This ensures that after iteration $i$, $\tau$ is a uniformly random selection from $\pi_1, \ldots, \pi_i$. 

Selecting  the better of two solutions with respect to the footrule distance in $O(n\log n)$ time is more involved. 
We describe how we can build a data structure that would allow achieving this. 

\subsection{The Rank Tree Data Structure} 
Recall that for a rank aggregation domain $\Pi=[\pi_1,\pi_2,\ldots, \pi_m]$ and  a complete ranking  $\pi$, each of size $n$, the Spearman's footrule distance $\mathcal{F}(\pi,\Pi)$ is defined as,   
\begin{align}\label{re:footrule}
\mathcal{F}(\pi,\Pi)&=\sum_{i=1}^{m}\mathcal{F}(\pi,\pi_i)=\sum_{i=1}^{m}\sum_{e\in U}|\pi(e)-\pi_i(e)|.
\end{align}
Even if we could afford to store all the rankings, 
a brute-force computation of $\mathcal{F}(\pi,\Pi)$ based on  \cref{re:footrule} takes $\Theta(mn)$ time. Note that although we can compute the exact value of $\mathcal{F}(\pi,\Pi)$ for a given $\pi$, we cannot provide the optimal ranking $\pi$ with minimum value for $\mathcal{F}(\pi,\Pi)$. Therefore, we use this method only to compare the costs of the rankings provided by the  LR-Aggregation and the Pick-A-Perm algorithms and choose the best one.

Here, we present a method for computing $\mathcal{F}(\pi,\Pi)$ in $O(n\log n)$ time and space. 
The last equality in \cref{re:footrule} can be rewritten as,
\begin{align}\label{re:foot}
\sum_{e\in U}\sum_{i=1}^{m}|\pi(e)-\pi_i(e)|
\end{align}
So, we can streamline the computation of $\mathcal{F}(\pi,\Pi)$ by computing the cumulative cost imposed by each element 
across all rankings in $\Pi$. 
Building on this insight, we develop the \textit{Rank Tree} data structure, 
which enables the computation of $\mathcal{F}(\pi,\Pi)$ in $O(n\log n)$ time.
Let $X_e=\{i: \pi_i(e)>\pi(e)\}$ and $Y_e=\{i: \pi_i(e)<\pi(e)\}$.  According to \cref{re:foot},  $\mathcal{F}(\pi,\Pi)=\sum_{e\in U}C(e)$, in which 
\begin{equation}
\begin{aligned}\label{rel:foot}
C(e)=\sum_{i=1}^{m}|\pi(e)-\pi_i(e)|&=\sum_{i\in X_e}(\pi_i(e)-\pi(e))+\sum_{i\in Y_e}(\pi(e)-\pi_i(e))\\&=\sum_{i\in X_e}\pi_i(e)-\sum_{i\in Y_e}\pi_i(e)+\pi(e)(|Y_e|-|X_e|)
\end{aligned}
\end{equation}
The Rank Tree data structure,  denoted by $\mathcal{R}(e)$, is constructed for each element $e$ to efficiently compute the terms in the last equality above. 
In this structure, a balanced binary search tree  is built on the numbers $1,2,\ldots, n$. 
Each node $u$  in $\mathcal{R}(e)$ is associated with  three attributes:  
\textit{repetition}, \textit{subtree size}, and \textit{subtree sum}, denoted as $re(u)$, $size(u)$, and $sum(u)$, respectively.

The Rank Tree is initialized with all node  attributes set to zero.
Upon receiving a new ranking $\pi_i$, the corresponding $\mathcal{R}(e)$ for each element $e$ 
is updated by traversing from the root to the node labeled with $\pi_i(e)$, 
which represents the position of $e$ in $\pi_i$, incrementing the $size$ attribute by 1 and the $sum$ attribute by $\pi_i(e)$ 
 at each node along the path. 
Additionally, the $re$ attribute of the target node is incremented by 1.
In other words, each time an element appears at a new position $j$ (after receiving a new input ranking),  
 the repetition attribute of the existing node for $j$ is increased. 
Subsequently,  the size and sum attributes for this node and other affected nodes in the tree are updated accordingly. 
This structure enables efficient computation of the cumulative cost of each element $e$ over all received rankings.
 \cref{fig:3} and \cref{fig:4} illustrate, as an example for  $\Pi=[\pi_1,\pi_2]$,   the updated rank trees of elements $A,C$, and $D$, after receiving the  rankings $\pi_1=\prec A, B, C, D\succ$ and  $\pi_2=\prec B,D,A,C\succ$, respectively.  
 Note that the properties $re, size$, and $sum$ are shown in  rectangles close to each node from left to right.
\begin{figure}[H]
\centering
\begin{subfigure}{.3\textwidth}
  \centering
  \includegraphics[width=.9\linewidth]{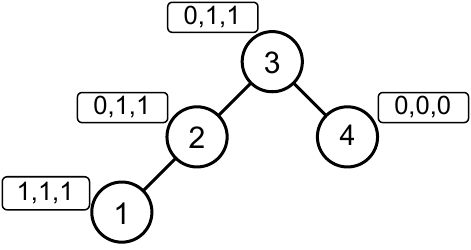} 
  \caption*{$\mathcal{R}(A)$}
\end{subfigure}%
\begin{subfigure}{.3\textwidth}
  \centering
  \includegraphics[width=.9\linewidth]{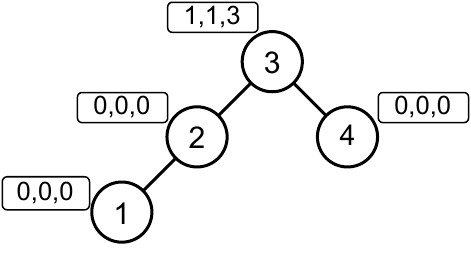} 
  \caption*{$\mathcal{R}(C)$}
\end{subfigure}%
\begin{subfigure}{.3\textwidth}
  \centering
  \includegraphics[width=.9\linewidth]{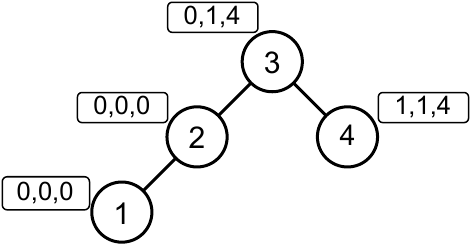} 
  \caption*{$\mathcal{R}(D)$}
\end{subfigure}
\caption{The rank trees of elements $A$, $C$, and $D$ correspond to  $\Pi=[\prec A, B, C, D\succ]$}
\label{fig:3}
\end{figure}

\begin{figure}[H]
\centering
\begin{subfigure}{.3\textwidth}
  \centering
  \includegraphics[width=.9\linewidth]{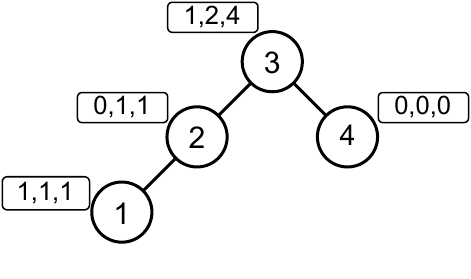} 
  \caption*{$\mathcal{R}(A)$}
\end{subfigure}%
\begin{subfigure}{.3\textwidth}
  \centering
  \includegraphics[width=.9\linewidth]{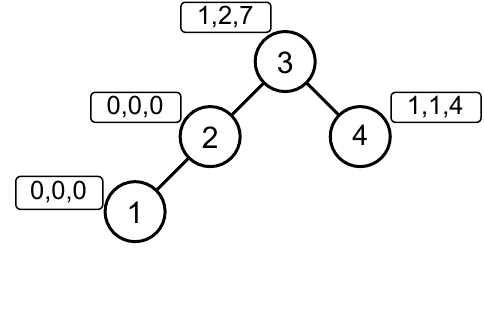} 
  \caption*{$\mathcal{R}(C)$}
\end{subfigure}%
\begin{subfigure}{.3\textwidth}
  \centering
  \includegraphics[width=.9\linewidth]{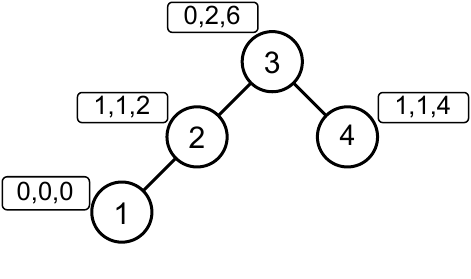} 
  \caption*{$\mathcal{R}(D)$}
\end{subfigure}
\caption{The rank trees of elements $A$, $C$, and $D$ correspond to  $\Pi=[\prec A, B, C, D\succ, \prec B,D,A,C\succ]$}
\label{fig:4}
\end{figure}
To update the set of rank trees  after receiving a ranking $\pi=\prec e_1,\ldots,e_n\succ$, we execute  \cref{alg:updateranktree} for each element $e_i$.  
\begin{algorithm}[H]
\caption{\textsc{Update Rank Tree}$(v, j)$}
\label{alg:updateranktree}
\begin{algorithmic}[1]
\Require Root $v$ of a rank tree, and an integer $j$ (representing the rank of the corresponding element)
\While{$v.value \ne j$}
    \State $size(v) \gets size(v) + 1$
    \State $sum(v) \gets sum(v) + j$
    \If{$v.value < j$}
        \State $v \gets rc(v)$
    \Else
        \State $v \gets lc(v)$
    \EndIf
\EndWhile
\State $re(v) \gets re(v) + 1$
\State $size(v) \gets size(v) + 1$
\State $sum(v) \gets sum(v) + j$
\end{algorithmic}
\end{algorithm}
\begin{theorem}\label{thm:updateranktree}
The set of  rank trees data structure for a rank aggregation domain $\Pi = [\pi_1, \ldots, \pi_m] $ over a set of $n$ elements uses $O(n^2)$ storage, and it is updated in $O(n \log n)$ time after receiving a  ranking $\pi$. 
\end{theorem}

\begin{theorem}\label{thm:cost}
Given a rank aggregation domain $\Pi = [\pi_1, . . . , \pi_m]$ over $n$ elements, the rank tree data structure enables the computation of $\mathcal{F}(\pi,\Pi)$ for a given  ranking $\pi$ in $O(n \log n)$ time. 
\end{theorem}
\begin{proof}
For each element $e$ in $\pi$,  the cost contribution of $e$ in $\mathcal{F}(\pi,\Pi)$ is computed using \cref{alg:cost} in $O(\log n)$ time. Note that this algorithm, during its traversal from the root of $\mathcal{R}(e)$ to the node with value $\pi(e)$, computes the values of terms $\sum_{i\in X_e}\pi_i(e)$,  $\sum_{i\in Y_e}\pi_i(e)$, $|Y_e|$, and $|X_e|$ of \cref{rel:foot}. These terms are denoted in \cref{alg:cost} by $c_{11}$, $c_{21}$, $c_{22}$, and $c_{12}$, respectively. Thus, the total time for computing $\mathcal{F}(\pi,\Pi)$ is $O(n\log n)$. 
\end{proof}

\begin{algorithm}[H]
\caption{\textsc{Cost}$(v, j)$} 
\label{alg:cost}
\begin{algorithmic}[1]
\Require Root $v$ of a rank tree, and integer $j$ (representing the rank of the corresponding element) 
\Ensure Cost contribution of the element with rank $j$
\State $c_{11}, c_{12}, c_{21}, c_{22} \gets 0$
\While{$v.value \ne j$}
    \If{$v.value < j$}
        \State $c_{21} \gets c_{21} + re(v) * v.value + sum(lc(v))$
        \State $c_{22} \gets c_{22} + re(v) + size(lc(v))$
        \State $v \gets rc(v)$
    \Else
        \State $c_{11} \gets c_{11} + re(v)* v.value + sum(rc(v))$
        \State $c_{12} \gets c_{12} + re(v) + size(rc(v))$
        \State $v \gets lc(v)$
    \EndIf
\EndWhile
\State $c_{11} \gets c_{11} + sum(rc(v))$
\State $c_{12} \gets c_{12} + size(rc(v))$
\State $c_{21} \gets c_{21} + sum(lc(v))$
\State $c_{22} \gets c_{22} + size(lc(v))$
\State\Return $(c_{11} - c_{21}) + j * (c_{22} - c_{12})$
\end{algorithmic}
\end{algorithm}
\cref{alg:dynamic_pap} introduces a dynamic method to pick a uniform random ranking, and \cref{alg:best_of} combines the Pick-A-Perm and LR-Aggregation methods.
Let us consider $\omega_{\scriptscriptstyle LR}$ and  $\omega_{\scriptscriptstyle PAP}$ as the aggregations obtained by  LR-Aggregation and dynamic Pick-A-Perm, respectively. To efficiently determine the best of LR-Aggregation and Pick-A-Perm, we employ \cref{alg:footdist} to compute the Spearman's footrule distances $\mathcal{F}(\omega_{\scriptscriptstyle LR},\Pi)$ and $\mathcal{F}(\omega_{\scriptscriptstyle PAP},\Pi)$. The aggregation that minimizes this distance is then returned as aggregated ranking. This algorithm  takes as input the array \textit{Rankroots} $\mathbf{R}_{rank}$, which stores pointers to the roots of rank trees, enabling efficient access.
It also uses \cref{alg:cost} as a sub-routine to compute the cost contribution of each element.
\begin{algorithm}[H]
\caption{\textsc{Spearman Footrule Distance}$(\pi, \mathbf{R}_{rank})$}
\label{alg:footdist}
\begin{algorithmic}[1]
\Require Ranking $\pi$, and Rankroots array $\mathbf{R}_{rank}$
\Ensure The Spearman's footrule distance between $\pi$ and $\Pi$
\State $d \gets 0$
\State $n \gets$ size of $\pi$
\For{$j \gets 1$ to $n$}
    \State $e \gets$ $j$-th element of $\pi$ \Comment{element with rank $j$ in  $\pi$}
    \State $d \gets d + \textsc{Cost}(\mathbf{R}_{rank}[e], j)$
\EndFor
\State\Return $d$
\end{algorithmic}
\end{algorithm}

\begin{algorithm}[H]
\caption{\textsc{Dynamic Pick-A-Perm}$(\omega_{\scriptscriptstyle PAP}, \pi, m)$}
\label{alg:dynamic_pap}
\begin{algorithmic}[1]
\Require Aggregated ranking $\omega_{\scriptscriptstyle PAP}$ of $\{\pi_1, \ldots, \pi_{m-1}\}$,  ranking $\pi$, total number of rankings $m$
\Ensure A uniformly selected ranking from $\{\pi_1, \ldots, \pi_m =\pi\}$
\State $p \gets$ uniformly random integer from $1$ to $m$
\If{$p = m$}
    \State $\omega_{\scriptscriptstyle PAP} \gets \pi$
\EndIf
\State\Return $\omega_{\scriptscriptstyle PAP}$
\end{algorithmic}
\end{algorithm}

\begin{algorithm}[H]
\caption{\textsc{Dynamic Rank Aggregation}$(\omega_{\scriptscriptstyle PAP}, \omega_{\scriptscriptstyle LR}, \pi, m, \mathbf{R}_{LR}, \mathbf{R}_{rank})$} 
\label{alg:best_of}
\begin{algorithmic}[1]
\Require Pick-A-Perm aggregated ranking $\omega_{\scriptscriptstyle PAP}$, LR-aggregated ranking $\omega_{\scriptscriptstyle LR}$, ranking $\pi$, number of rankings $m$, LRroots array $\mathbf{R}_{LR}$, and Rankroots array $\mathbf{R}_{rank}$
\Ensure The better of $\omega_{\scriptscriptstyle PAP}$ and $\omega_{\scriptscriptstyle LR}$ based on Spearman's footrule distance

\State $\omega_{\scriptscriptstyle LR} \gets$ \textsc{LR-Aggregation}$(\omega_{\scriptscriptstyle LR}, \pi, \mathbf{R}_{LR})$
\State $n\gets$ size of ranking $\pi$
\For{$j \gets 1$ to $n$}
    \State $e \gets$ $j$-th element of $\pi$
    \State $v \gets \mathbf{R}_{rank}[e]$
    \State \textsc{Update Rank Tree}$(v, j)$
\EndFor
\State $\omega_{\scriptscriptstyle PAP} \gets$ \textsc{Dynamic Pick-A-Perm}$(\omega_{\scriptscriptstyle PAP}, \pi, m)$
\State $score_{\scriptscriptstyle LR} \gets$ \textsc{Spearman Footrule Distance}$(\omega_{\scriptscriptstyle LR}, \mathbf{R}_{rank})$
\State $score_{\scriptscriptstyle PAP} \gets$ \textsc{Spearman Footrule Distance}$(\omega_{\scriptscriptstyle PAP}, \mathbf{R}_{rank})$
\If{$score_{\scriptscriptstyle LR} < score_{\scriptscriptstyle PAP}$}
    \State $\omega_{best} \gets \omega_{\scriptscriptstyle LR}$
\Else
    \State $\omega_{best} \gets \omega_{\scriptscriptstyle PAP}$
\EndIf
\State\Return $\omega_{best}$
\end{algorithmic}
\end{algorithm}
\section{Experimental Results}\label{sec:experiment}
This section evaluates the quality of the solutions produced by  the  LR-Aggregation algorithm, 
using experiments on synthetic and real-world datasets. 
We compare the algorithm against the optimal footrule aggregation, the Pick-A-Perm, and the Best-Input-Ranking methods.  

For better comparison, we computed three key ratios:
\begin{align*}
\alpha_{\scriptscriptstyle LR}=\frac{\mathcal{F}(\pi,\Pi)}{\mathcal{F}(\pi^*,\Pi)}, \quad
\alpha_{\scriptscriptstyle PAP}=\frac{\mathcal{F}(\pi^\prime,\Pi)}{\mathcal{F}(\pi^*, \Pi)},\quad
\alpha_{\scriptscriptstyle BIR}=\frac{\mathcal{F}(\pi^{\prime\prime},\Pi)}{\mathcal{F}(\pi^*,\Pi)}  
\end{align*}
where $\pi^*,\pi$, $\pi^\prime$ and $\pi''$ denote the optimal footrule aggregation, LR-Aggregation, Pick-A-Perm aggregation,  and the output of the Best-Input-Ranking--which selects the input ranking with the smallest footrule distance to the domain--respectively.
These ratios provide insight into how closely LR-Aggregation, Pick-A-Perm aggregation, and Best Input Ranking algorithm approximate the optimal aggregation. 
\subsection{Experiments on Random Datasets}
We evaluate the proposed method on synthetic rank aggregation domains,  generated using three models: (i) uniform random sampling,  
(ii) biased domain (described  below), and (iii) the Mallows model~\cite{mallows}. 
Each domain $\Pi$ consists of $m \in \{1000, 2000, \dots, 10000\}$ rankings,  each of size $n \in \{64, 128, 256, 512\}$. 
Across all experiments, we evaluated the final aggregated ranking obtained after receiving the $m$ input rankings. 
We use \cref{alg:dyrankagg} such that  the LR-trees contain the presence information of each element across the entire domain.

In the  uniform  random domain, each  ranking was sampled uniformly from all permutations of  $n$ elements; \crefrange{tab:64}{tab:512} present the experimental results for this domain. 
 \begin{table}[H]
\centering
\captionsetup{justification=raggedright,singlelinecheck=false,position=top}
\caption{Comparing aggregation methods on uniform domain, n=64}\label{tab:64}
\begin{adjustbox}{width=1\textwidth}
\small
\begin{tabular}{lllllllllll}
\toprule
$m$ & 1000 & 2000 & 3000 & 4000 & 5000 & 6000 & 7000 & 8000 & 9000 & 10000 \\
\midrule
Optimal aggregation &1333458 &2686850 &4048560 &5401060 &6756788 &8115854 &9483262 &10841720 &12197208 &13562750
 \\
LR-Aggregation  
&1334612 &2689444 &4051688 &5403818 &6758900 &8117388 &9486436 &10843846 &12199950 &13565574
\\

Pick-A-Perm &1366034 &2722980 &4085670 &5456432 &6833688 &8178098 &9549274 &10925314 &12269986 &13629332
\\

\midrule
$\alpha_{LR}$ &1.001 &1.001 &1.001 &1.001 &1.0 &1.0 &1.0 &1.0 &1.0 &1.0
  \\
$\alpha_{\scriptscriptstyle PAP}$ &1.024 &1.013 &1.009 &1.01 &1.011 &1.008 &1.007 &1.008 &1.006 &1.005
  \\
\bottomrule
\end{tabular}
\end{adjustbox}
\end{table}
\begin{table}[H]
\centering
\captionsetup{justification=raggedright,singlelinecheck=false,position=top}
\caption{Comparing aggregation methods on uniform domain, n=128}\label{tab:128}
\small
\resizebox{\textwidth}{!}{%
\begin{tabular}{lllllllllll}
\toprule
$m$ & 1000 & 2000 & 3000 & 4000 & 5000 & 6000 & 7000 & 8000 & 9000 & 10000 \\
\midrule
Optimal aggregation &5336204 &10724820 &16172742 &21597540 &27020404 &32461084 &37909314 &43356058 &48775678 &54226688
 \\
LR-Aggregation &5339892 &10729330 &16181194 &21604648 &27026330 &32467780 &37927022 &43373450 &48785666 &54244908  \\
Pick-A-Perm &5449568 &10926136 &16410910 &21862770 &27289842 &32781490 &38194364 &43664684 &49132498 &54583086
  \\
\midrule
$\alpha_{LR}$ &1.001 &1.0 &1.001 &1.0 &1.0 &1.0 &1.0 &1.0 &1.0 &1.0
  \\
$\alpha_{\scriptscriptstyle PAP}$ &1.021 &1.019 &1.015 &1.012 &1.01 &1.01 &1.008 &1.007 &1.007 &1.007  \\
\bottomrule
\end{tabular}
}
\end{table}
\begin{table}[H]
\centering
\captionsetup{justification=raggedright,singlelinecheck=false,width=\textwidth}
\caption{Comparing aggregation methods on uniform domain, $n=256$}\label{tab:256}
\small
\resizebox{\textwidth}{!}{%
\begin{tabular}{lllllllllll}
\toprule
$m$ & 1000 & 2000 & 3000 & 4000 & 5000 & 6000 & 7000 & 8000 & 9000 & 10000 \\
\midrule
Optimal aggregation &21364196 &42978692 &64688638 &86385780 &108105668 &129941788 &151610022 &173416762 &195182668 &216953088
 \\
LR-Aggregation &21378184 &42999152 &64729636 &86430458 &108147370 &129995492 &151659670 &173442520 &195241170 &217011630  \\
Pick-A-Perm &21813002 &43692424 &65616138 &87341642 &109182398 &131068444 &152877270 &174691440 &196543268 &218297896  \\
\midrule
$\alpha_{LR}$ &1.001 &1.0 &1.001 &1.001 &1.0 &1.0 &1.0 &1.0 &1.0 &1.0  \\
$\alpha_{\scriptscriptstyle PAP}$ &1.021 &1.017 &1.014 &1.011 &1.01 &1.009 &1.008 &1.007 &1.007 &1.006  \\
\bottomrule
\end{tabular}
}
\end{table}
\begin{table}[H]
\centering
\captionsetup{justification=raggedright,singlelinecheck=false,position=top}
\caption{Comparing aggregation methods on uniform domain, n=512}\label{tab:512}
\begin{adjustbox}{width=1\textwidth}
\small
\begin{tabular}{lllllllllll}
\toprule
$m$ & 1000 & 2000 & 3000 & 4000 & 5000 & 6000 & 7000 & 8000 & 9000 & 10000 \\
\midrule
Optimal aggregation &85382782 &171912688 &258868998 &345654312 &432565670 &519511610 &606715014 &693578060 &780831062 &867839784
 \\
LR-Aggregation &85460806 &172007138 &258986984 &345783456 &432720334 &519658944 &606887614 &693811870 &781058994 &868051882  \\
Pick-A-Perm &87357914 &174585464 &261964586 &349278268 &436861626 &524174852 &611797790 &698863312 &786094312 &873292744  \\
\midrule
$\alpha_{LR}$ &1.001 &1.001 &1.0 &1.0 &1.0 &1.0 &1.0 &1.0 &1.0 &1.0  \\
$\alpha_{\scriptscriptstyle PAP}$ &1.023 &1.016 &1.012 &1.01 &1.01 &1.009 &1.008 &1.008 &1.007 &1.00  \\
\bottomrule
\end{tabular}
\end{adjustbox}
\end{table}

Although the outputs of LR-Aggregation are superior to Pick-A-Perm, the difference is not significant. 
Intuitively, the reason is that when the domain of rankings is uniformly random, 
\textit{any} ranking would be a good approximation to the optimal aggregation ranking.
Hence the significance of the other two rank aggregation domains: the \textit{biased}, and the Mallows models.
In both the biased and Mallows domains, we also compared the results against the Best-Input-Ranking
algorithm. Note that the running
time of this algorithm is high and hence it is not suitable for the streaming setting; we use its results to show that even
if the random choice in Pick-A-Perm would magically coincide with the best input ranking, the result would still be
significantly inferior to our algorithm’s output.

In the biased  domain, we selected an arbitrary ranking as the base, and aligned  the order of $k$ random pairs (for different values of $k$) 
with the base ranking for the entire set of [randomly selected] rankings in the domain. 
Here, the intuition is that the base ranking plays the role of the optimal aggregated ranking, 
towards which each ranking in the domain is biased. The experimental outcomes for this domain are detailed in \crefrange{tabb:64}{tabb:512}.
\begin{table}[H]
\centering
\captionsetup{justification=raggedright,singlelinecheck=false,position=top}
\caption{Comparing aggregation methods on biased  domain, n=64}\label{tabb:64}
\begin{adjustbox}{width=1\textwidth}
\small
\begin{tabular}{lllllllllll}
\toprule
$m$ & 1000 & 2000 & 3000 & 4000 & 5000 & 6000 & 7000 & 8000 & 9000 & 10000 \\
\midrule
Optimal aggregation &743036 &1487292 &2239252 &2978586 &3725904 &4459232 &5217110 &5958516 &6689550 &7450008
 \\
LR-Aggregation &743036 &1487292 &2239252 &2978586 &3725904 &4459232 &5217110 &5958516 &6689550 &7450008  \\
Pick-A-Perm &971840 &1875682 &2965896 &3689616 &5349964 &5738160 &6723106 &7654212 &8605574 &9828204  \\
Best-Input-Ranking &855146 &1735202 &2604608 &3408994 &4299444 &5149104 &6040246 &6804482 &7715724 &8416922  \\
\midrule
$\alpha_{LR}$ &1.0 &1.0 &1.0 &1.0 &1.0 &1.0 &1.0 &1.0 &1.0 &1.0  \\
$\alpha_{\scriptscriptstyle PAP}$ &1.308 &1.261 &1.325 &1.239 &1.436 &1.287 &1.289 &1.285 &1.286 &1.319  \\
$\alpha_{BIR}$ &1.151 &1.167 &1.163 &1.145 &1.154 &1.155 &1.158 &1.142 &1.153 &1.13  \\
\bottomrule
\end{tabular}
\end{adjustbox}
\end{table}
\begin{table}[H]
\centering
\captionsetup{justification=raggedright,singlelinecheck=false,position=top}
\caption{Comparing aggregation methods on biased domain, n=128}\label{tabb:128}
\begin{adjustbox}{width=1\textwidth}
\small
\begin{tabular}{lllllllllll}
\toprule
$m$ & 1000 & 2000 & 3000 & 4000 & 5000 & 6000 & 7000 & 8000 & 9000 & 10000 \\
\midrule
Optimal aggregation &3033284 &6071348 &9097872 &12133560 &15175532 &18175410 &21264862 &24263942 &27292084 &30319972
 \\
LR-Aggregation &3033300 &6071348 &9097872 &12133560 &15175554 &18175410 &21264862 &24263942 &27292132 &30319972  \\
Pick-A-Perm &3938578 &8339116 &11739170 &15865000 &19169840 &23574446 &28333590 &31146058 &37126788 &39139142  \\
Best-Input-Ranking &3682784 &7304282 &10805192 &14401996 &18053338 &21512886 &25129174 &28536526 &32049500 &36111514  \\
\midrule
$\alpha_{LR}$ &1.0 &1.0 &1.0 &1.0 &1.0 &1.0 &1.0 &1.0 &1.0 &1.0  \\
$\alpha_{\scriptscriptstyle PAP}$ &1.298 &1.374 &1.29 &1.308 &1.263 &1.297 &1.332 &1.284 &1.36 &1.291  \\
$\alpha_{BIR}$ &1.214 &1.203 &1.188 &1.187 &1.19 &1.184 &1.182 &1.176 &1.174 &1.191  \\
\bottomrule
\end{tabular}
\end{adjustbox}
\end{table}
\begin{table}[H]
\centering
\captionsetup{justification=raggedright,singlelinecheck=false,position=top}
\caption{Comparing aggregation methods on biased domain, n=256}\label{tabb:256}
\begin{adjustbox}{width=1\textwidth}
\small
\begin{tabular}{lllllllllll}
\toprule
$m$ & 1000 & 2000 & 3000 & 4000 & 5000 & 6000 & 7000 & 8000 & 9000 & 10000 \\
\midrule
Optimal aggregation &12253932 &24452568 &36743086 &48860942 &61107774 &73281364 &85569474 &97790564 &110070702 &122391202
 \\
LR-Aggregation &12253948 &24452696 &36743176 &48861046 &61107800 &73281368 &85569474 &97790564 &110070740 &122391202  \\
Pick-A-Perm &15947806 &31613332 &47726560 &64908026 &79571614 &97173622 &112373732 &132202652 &143608576 &157744964  \\

Best-Input-Ranking &15083756 &30132186 &44134544 &59912854 &75280060 &89402594 &103568134 &119194542 &133732808 &148359130\\
\midrule
$\alpha_{LR}$ &1.0 &1.0 &1.0 &1.0 &1.0 &1.0 &1.0 &1.0 &1.0 &1.0  \\
$\alpha_{\scriptscriptstyle PAP}$ &1.301 &1.293 &1.299 &1.328 &1.302 &1.326 &1.313 &1.352 &1.305 &1.289  \\
$\alpha_{BIR}$&1.231 &1.232 &1.201 &1.226 &1.232 &1.22 &1.21 &1.219 &1.215 &1.212\\
\bottomrule
\end{tabular}
\end{adjustbox}
\end{table}

\begin{table}[H]
\centering
\captionsetup{justification=raggedright,singlelinecheck=false,position=top}
\caption{Comparing aggregation methods on biased domain, n=512}\label{tabb:512}
\begin{adjustbox}{width=1\textwidth}
\small
\begin{tabular}{lllllllllll}
\toprule
$m$ & 1000 & 2000 & 3000 & 4000 & 5000 & 6000 & 7000 & 8000 & 9000 & 10000 \\
\midrule
Optimal aggregation &49143766 &98275564 &147516688 &196647092 &245773046 &294961978 &344362278 &393714422 &442725918 &491860308
 \\
LR-Aggregation &49143936 &98275676 &147516738 &196647272 &245773160 &294962080 &344362376 &393714566 &442726004 &491860504  \\
Pick-A-Perm &64890874 &128957704 &198219264 &260797558 &322954020 &382025586 &458634100 &515974814 &601390482 &646653372 \\
Best-Input-Ranking &61817908 &121979318 &183983100 &245768570 &307556470 &366757806 &429832440 &488255222 &544648576 &609560198\\
\midrule
$\alpha_{LR}$ &1.0 &1.0 &1.0 &1.0 &1.0 &1.0 &1.0 &1.0 &1.0 &1.0  \\
$\alpha_{\scriptscriptstyle PAP}$ &1.32 &1.312 &1.344 &1.326 &1.314 &1.295 &1.332 &1.311 &1.358 &1.315   \\
$\alpha_{BIR}$&1.258 &1.241 &1.247 &1.25 &1.251 &1.243 &1.248 &1.24 &1.23 &1.239\\
\bottomrule
\end{tabular}
\end{adjustbox}
\end{table}
\crefrange{tabbb:64}{tabbb:512} shows the results of evaluation using the  Mallows model, 
a probabilistic framework that produces rankings centered around a modal (reference) ranking $\pi$~\cite{mallows}.
\vspace{-1pt}
As can be seen from these tables, the LR-Aggregation algorithm performs significantly better than both Pick-A-Perm, 
and the Best-Input-Ranking algorithm.
\begin{table}[H]
\centering
\captionsetup{justification=raggedright,singlelinecheck=false,position=top}
\caption{Comparing aggregation methods on Mallows domain, n=64}\label{tabbb:64}
\begin{adjustbox}{width=1\textwidth}
\small
\begin{tabular}{lllllllllll}
\toprule
$m$ & 1000 & 2000 & 3000 & 4000 & 5000 & 6000 & 7000 & 8000 & 9000 & 10000 \\
\midrule
Optimal aggregation &1058248 &2125736 &3196684 &4250218 &5295542 &6373098 &7421398 &8505328 &9525466 &10657240

 \\
LR-Aggregation  &1064202 &2135624 &3201178 &4266880 &5310938 &6390870 &7443418 &8523896 &9535502 &10688846\\
Pick-A-Perm &1183684 &2471510 &3746628 &4730992 &5852518 &6754998 &8163966 &9211450 &10000384 &11913640 \\
Best-Input-Ranking  &1095172 &2196802 &3283122 &4391372 &5466274 &6570268 &7648926 &8747674 &9788980 &10936558 \\
\midrule
$\alpha_{LR}$ &1.006 &1.005 &1.001 &1.004 &1.003 &1.003 &1.003 &1.002 &1.001 &1.003\\
$\alpha_{\scriptscriptstyle PAP}$ &1.119 &1.163 &1.172 &1.113 &1.105 &1.06 &1.1 &1.083 &1.05 &1.118 \\
$\alpha_{BIR}$ &1.035 &1.033 &1.027 &1.033 &1.032 &1.031 &1.031 &1.028 &1.028 &1.026  \\
\bottomrule
\end{tabular}
\end{adjustbox}
\end{table}

\begin{table}[H]
\centering
\captionsetup{justification=raggedright,singlelinecheck=false,position=top}
\caption{Comparing aggregation methods on Mallows  domain, n=128}\label{tabbb:128}
\begin{adjustbox}{width=1\textwidth}
\small
\begin{tabular}{lllllllllll}
\toprule
$m$ & 1000 & 2000 & 3000 & 4000 & 5000 & 6000 & 7000 & 8000 & 9000 & 10000 \\
\midrule
Optimal aggregation &4308840 &8645888 &12888628 &17249580 &21543302 &25869410 &30201792 &34431772 &38842100 &43157880 
 \\
LR-Aggregation &4332376 &8685294 &12947578 &17308464 &21632840 &25940466 &30343624 &34558106 &38973366 &43298942 \\
Pick-A-Perm &4730716 &9285498 &14401970 &19099006 &24939830 &27563490 &36495108 &38334302 &42779086 &46707656  \\
Best-Input-Ranking &4473392 &8957822 &13363412 &17811722 &22340052 &26817254 &31251154 &35591058 &40159234 &44673870  \\
\midrule
$\alpha_{LR}$ &1.005 &1.005 &1.005 &1.003 &1.004 &1.003 &1.005 &1.004 &1.003 &1.003 \\
$\alpha_{\scriptscriptstyle PAP}$ &1.098 &1.074 &1.117 &1.107 &1.158 &1.065 &1.208 &1.113 &1.101 &1.082 \\
$\alpha_{BIR}$ &1.038 &1.036 &1.037 &1.033 &1.037 &1.037 &1.035 &1.034 &1.034 &1.035  \\
\bottomrule
\end{tabular}
\end{adjustbox}
\end{table}

\begin{table}[H]
\centering
\captionsetup{justification=raggedright,singlelinecheck=false,position=top}
\caption{Comparing aggregation methods on Mallows  domain, n=256}\label{tabbb:256}
\begin{adjustbox}{width=1\textwidth}
\small
\begin{tabular}{lllllllllll}
\toprule
$m$ & 1000 & 2000 & 3000 & 4000 & 5000 & 6000 & 7000 & 8000 & 9000 & 10000 \\
\midrule
Optimal aggregation &17081480 &34277148 &51629010 &68705376 &86191660 &103090112 &120651510 &138143926 &155168188 &172153242
 \\
LR-Aggregation &17136012 &34415204 &51831918 &68978534 &86567716 &103322648 &121041802 &138503272 &155610616 &172576514  \\
Pick-A-Perm  &19540024 &37020352 &55648176 &76948016 &92154440 &108931520 &129229254 &151240918 &164902270 &202530498\\
Best-Input-Ranking  &17859906 &35790424 &53992090 &71556914 &89946686 &107943338 &125938368 &144253548 &162233692 &179590808\\
\midrule
$\alpha_{LR}$ &1.003 &1.004 &1.004 &1.004 &1.004 &1.002 &1.003 &1.003 &1.003 &1.002
 \\
$\alpha_{\scriptscriptstyle PAP}$ &1.144 &1.08 &1.078 &1.12 &1.069 &1.057 &1.071 &1.095 &1.063 &1.176
 \\
$\alpha_{BIR}$ &1.046 &1.044 &1.046 &1.042 &1.044 &1.047 &1.044 &1.044 &1.046 &1.043 \\
\bottomrule
\end{tabular}
\end{adjustbox}
\end{table}

\begin{table}[H]
\centering
\captionsetup{justification=raggedright,singlelinecheck=false,position=top}
\caption{Comparing aggregation methods on Mallows domain, n=512}\label{tabbb:512}
\begin{adjustbox}{width=1\textwidth}
\small
\begin{tabular}{lllllllllll}
\toprule
$m$ & 1000 & 2000 & 3000 & 4000 & 5000 & 6000 & 7000 & 8000 & 9000 & 10000 \\
\midrule
Optimal aggregation &68487690 &137679962 &206444602 &275704358 &344698154 &620711396 &482384060 &551471928 &621579330 &689884924
 \\
LR-Aggregation &68763394 &138294264 &207253302 &277292366 &346522820 &624068186 &485342600 &554272940 &624941784 &694008872    \\
Pick-A-Perm  &73427898 &146446548 &222891236 &298778334 &370935304 &654002326 &540112500 &596962984 &677280708 &733920194   \\
Best-Input-Ranking  &72145306 &144572252 &216688240 &288711516 &361367662 &648990014 &505491488 &577932814 &651220456 &723958350  \\
\midrule
$\alpha_{LR}$ &1.004 &1.004 &1.004 &1.006 &1.005 &1.005 &1.006 &1.005 &1.005 &1.006
 \\
$\alpha_{\scriptscriptstyle PAP}$ &1.072 &1.064 &1.08 &1.084 &1.076 &1.054 &1.12 &1.082 &1.09 &1.064
 \\
$\alpha_{BIR}$ &1.053 &1.05 &1.05 &1.047 &1.048 &1.046 &1.048 &1.048 &1.048 &1.049\\
\bottomrule
\end{tabular}
\end{adjustbox}
\end{table}

\subsection{Experiment on a Real Dataset}
This section presents the results of applying the LR-Aggregation algorithm to a real-world dataset, demonstrating  its practical effectiveness.  The dataset includes the   academic performance of 32 students across 10 distinct lessons, with each student  assigned a corresponding grade. The  ordering of students based on their grades in each lesson yields an individual ranking for that lesson. These rankings, denoted as $\pi_1,\ldots,\pi_{10}$, are listed in the following.
\begin{align*}
&\pi_1=\prec k, s, t, C, Q, p, c, i, S, A, l, H, J, h, R, M, U, B, D, L, I, F, n, P, j, G, T, o, m, u, b, r \succ\\
&\pi_2=\prec G, D, U, M, J, b, F, k, R, l, c, A, I, S, C, p, t, o, Q, H, u, s, i, r, T, P, n, h, L, m, j, B \succ\\
&\pi_3=\prec c, l, H, C, k, M, S, p, n, J, s, P, G, o, B, u, t, h, U, L, I, F, Q, m, A, j, T, i, R, D, b, r \succ\\
&\pi_4=\prec H, C, R, k, M, J, s, G, U, b, l, n, j, c, p, P, u, t, D, F, Q, A, o, L, S, h, m, i, B, I, T, r \succ\\
&\pi_5=\prec C, H, p, R, k, A, j, u, c, Q, S, J, U, h, l, r, n, o, D, T, P, b, t, i, M, B, s, G, F, I, m, L \succ\\
&\pi_6=\prec l, S, j, c, Q, U, L, R, A, F, n, B, p, P, k, o, H, C, M, J, h, s, D, m, t, i, I, b, T, r, G, u \succ\\
&\pi_7=\prec c, H, l, n, k, Q, h, p, A, P, R, s, b, C, M, r, u, S, j, B, J, U, o, I, m, D, i, G, T, F, t, L \succ\\
&\pi_8=\prec c, l, H, C, R, n, J, h, G, b, F, D, m, Q, M, S, j, P, I, L, p, r, A, s, t, o, k, B, T, i, U, u \succ\\
&\pi_9=\prec H, M, b, G, n, c, R, k, P, l, C, Q, S, h, A, s, p, t, u, B, I, U, F, j, L, r, J, i, T, o, m, D \succ\\
&\pi_{10}=\prec c, l, H, R, k, n, J, j, M, h, b, p, C, G, t, S, i, Q, L, T, o, F, A, s, P, U, D, r, B, u, I, m \succ\\
\end{align*}
The details of the data, including student IDs and grades for each lesson, are provided in the Appendix  to support a thorough analysis of the experimental results. 

Our experiments involve  three  aggregation methods: Spearman footrule aggregation,  LR-Aggregation, and average aggregation (sorting by their GPAs). 

 \cref{tab:evaluation_results}  presents the rankings obtained by these methods.
As the experimental results indicate, there are instances where students with higher GPAs (average aggregation)  are positioned lower than individuals with lower averages in the rankings based on Spearman aggregation and  LR-Aggregation.  For example, student $F$ with a  greater GPA than student $t$, is ranked lower than $t$ in
 the ordering obtained from two other aggregations. This occurs because $t$ has outperformed $F$ in more courses. A similar pattern is observed for students $n$ and $Q$. 
 
These findings highlight the nuances of different aggregation methods and their implications on ranking students based on various criteria. 

\begin{table}[H]
\centering
\captionsetup[table]{justification=raggedright,singlelinecheck=false,position=top}
\caption{Results of different aggregation methods for a real dataset}\label{tab:evaluation_results}
\begin{adjustbox}{width=0.8\textwidth}
\small
\begin{tabular}{lllllllllll}
\toprule
\textbf{Aggregation Method} & \textbf{Aggregated ranking} & \textbf{Footrule distance} \\
\midrule
\textbf{Spearman Aggregation} &  clHCkMRnJpbQShAGPtjUsFoILDiBTumr &1862 \\
\textbf{LR-Aggregation} & cHlCkRMpnJbQGShAPtUjsFoDLIBiTumr &1862 \\
\textbf{Average Aggregation} & clHCRkQMSpnJAhjsPGUbFtoDLIBiuTmr&1924 \\
\bottomrule
\end{tabular}
\end{adjustbox}
\end{table}
\section{Discussion and Conclusion}\label{sec:conclusion}
The Dynamic Rank Aggregation algorithm devised in this article is built on top of several ideas. 
It utilizes a novel look on the Spearman footrule distance, which forms the foundation for the LR-Aggregation algorithm. 
The algorithm, in turn, utilizes the LR-tree data structure, which allows the efficient calculation of the number of occurrences of 
any element in a subinterval. 

Another important idea is that the well-known Pick-A-Perm algorithm can be adapted to the dynamic setting. 
This required showing the approximation guarantee of 2, plus the observation that
the algorithm could be efficiently implemented in the dynamic setting through reservoir sampling. 

Experimental evaluations show that, not only the final Dynamic Rank Aggregation algorithm (\cref{alg:best_of}) is very efficient in practice, 
but also crucially it is much more efficient than Pick-A-Perm alone. 
This means combining the two algorithms, which we show can be done efficiently in a dynamic setting,
is crucial for getting excellent practical efficiency, on top of the theoretical [expected] approximation guarantee
of two by Pick-A-Perm. 
The algorithm could also be utilized in the static setting, where it has a running time of $O(mn\log n)$ on 
a domain of $m$ rankings. 
Also the memory used by the algorithm is $O(n^2)$, making it useful in streaming settings, where storing all input rankings is impractical.
Although the primary measure used in this work is the Spearman footrule distance, future research could explore
alternative performance metrics such as Kendall-tau or machine learning-based distance measures. These metrics may
offer deeper insights into ranking discrepancies and could reveal cases where the LR-Aggregation algorithm performs
even more effectively.

Beyond handling  incoming rankings, the proposed LR-Aggregation algorithm is also capable of managing incremental updates to previous rankings. This includes scenarios where earlier rankings are modified, such as in user feedback systems or evolving preferences. The underlying LR-tree structure supports efficient local adjustments, enabling the aggregated ranking to reflect such changes without full recomputation. While this flexibility expands the applicability of the algorithm to more complex dynamic environments, additional work is needed to fully optimize and formalize this capability. 
We identify this as an important direction for future research, aimed at further strengthening the adaptability and efficiency of our dynamic rank aggregation framework.

This paper has started a new research area on the classic well-known rank aggregation problem focusing on new  applications. Key directions for future work include:
\begin{itemize}
    \item Reducing the storage requirements of the LR-tree data structure while maintaining its efficient update capabilities.   
    \item Investigating  dynamic rank aggregation involving partial rankings.
    \item Explore incremental rank aggregation, particularly in scenarios where previous rankings themselves are modified. While our algorithm can handle new rankings incrementally, updating the aggregated ranking when previous rankings are altered is an important extension that requires further research.  
    \item More theoretical investigation of the LR-Aggregation algorithm, in particular proving a good approximation factor.
    \item Generalizing the approach to other distance measures, beyond Spearman's footrule, such as Kendall-tau distance.
    \item Applying and adapting the proposed method to real-world applications such as recommendation systems or distributed search results.    
\end{itemize}


\printbibliography
\section{Appendix}
\begin{table}[H]
    \centering
    \caption{Grades of students in 10 lessons}
    \label{tab:mytable}
    \begin{tabular}{ccccccccccc}
        \hline
        \textbf{} & $1$ & $2$ & $3$& $4$  & $5$ & $6$ & $7$ & $8$ & $9$ &$10$ \\
        \hline
        \textbf{A} & 19.20 & 17.80 & 18.50 & 17.00 & 18.00 & 18.00 & 18.60 & 17.00 & 17.40 & 16.40 \\
        \textbf{B} & 17.30 & 10.20 & 20.00 & 14.00 & 14.50 & 17.40 & 16.60 & 14.00 & 16.20 & 13.60 \\
        \textbf{b} & 10.50 & 19.00 & 15.90 & 20.00 & 15.00 & 13.20 & 17.70 & 20.00 & 19.60 & 18.80 \\
        \textbf{C} & 19.70 & 17.50 & 20.00 & 20.00 & 20.00 & 15.50 & 17.50 & 20.00 & 18.50 & 18.40 \\
        \textbf{c} & 19.40 & 17.80 & 20.00 & 18.50 & 17.50 & 19.80 & 20.00 & 20.00 & 18.90 & 20.00 \\
        \textbf{D} & 16.90 & 20.00 & 16.00 & 17.70 & 15.50 & 15.00 & 13.90 & 20.00 & 11.10 & 14.50 \\
        \textbf{F} & 16.70 & 19.00 & 18.80 & 17.50 & 13.50 & 18.00 & 13.00 & 20.00 & 15.40 & 16.80 \\
        \textbf{G} & 16.00 & 20.00 & 20.00 & 20.00 & 14.00 & 10.00 & 13.70 & 20.00 & 19.50 & 18.40 \\
        \textbf{H} & 18.60 & 16.40 & 20.00 & 20.00 & 19.50 & 16.00 & 20.00 & 20.00 & 20.00 & 20.00 \\
        \textbf{h} & 18.00 & 14.00 & 19.00 & 16.20 & 16.50 & 15.00 & 19.30 & 20.00 & 17.70 & 19.00 \\
        \textbf{I} & 16.90 & 17.70 & 19.00 & 12.80 & 13.50 & 13.50 & 16.00 & 19.00 & 16.10 & 12.00 \\
        \textbf{i} & 19.30 & 15.60 & 16.50 & 15.50 & 15.00 & 14.00 & 13.80 & 12.50 & 13.60 & 17.40 \\
        \textbf{J} & 18.60 & 19.40 & 20.00 & 20.00 & 17.50 & 15.00 & 16.10 & 20.00 & 14.20 & 19.20 \\
        \textbf{j} & 16.00 & 12.10 & 17.60 & 19.00 & 18.00 & 20.00 & 17.00 & 19.50 & 15.00 & 19.20 \\
        \textbf{k} & 20.00 & 18.70 & 20.00 & 20.00 & 18.50 & 16.30 & 19.50 & 14.00 & 18.80 & 20.00 \\
        \textbf{L} & 16.90 & 13.40 & 19.00 & 17.00 & 12.00 & 19.00 & 10.00 & 18.50 & 14.80 & 17.20 \\
        \textbf{l} & 18.90 & 18.10 & 20.00 & 19.50 & 16.00 & 20.00 & 19.70 & 20.00 & 18.50 & 20.00 \\
        \textbf{M} & 17.90 & 19.40 & 20.00 & 20.00 & 14.50 & 15.50 & 17.50 & 19.50 & 20.00 & 19.10 \\
        \textbf{m} & 12.70 & 12.30 & 18.60 & 16.00 & 13.00 & 14.50 & 15.70 & 20.00 & 12.20 & 10.00 \\
        \textbf{n} & 16.40 & 14.50 & 20.00 & 19.50 & 15.50 & 17.50 & 19.60 & 20.00 & 19.20 & 20.00 \\
        \textbf{o} & 13.70 & 17.10 & 20.00 & 17.00 & 15.50 & 16.20 & 16.00 & 16.00 & 13.00 & 17.10 \\
        \textbf{P} & 16.20 & 15.20 & 20.00 & 18.00 & 15.00 & 16.50 & 18.50 & 19.50 & 18.60 & 15.10 \\
        \textbf{p} & 19.50 & 17.40 & 20.00 & 18.10 & 19.00 & 16.60 & 18.70 & 18.00 & 17.30 & 18.50 \\
        \textbf{Q} & 19.60 & 16.90 & 18.60 & 17.20 & 17.50 & 19.50 & 19.50 & 19.50 & 18.50 & 17.20 \\
        \textbf{R} & 17.90 & 18.50 & 16.00 & 20.00 & 18.50 & 18.00 & 18.30 & 20.00 & 18.80 & 20.00 \\
        \textbf{r} & 10.00 & 15.60 & 9.00 & 10.00 & 16.00 & 12.00 & 17.20 & 18.00 & 14.70 & 14.50 \\
        \textbf{S} & 19.20 & 17.60 & 20.00 & 16.20 & 17.50 & 20.00 & 17.00 & 19.50 & 18.50 & 17.90 \\
        \textbf{s} & 20.00 & 15.70 & 20.00 & 20.00 & 14.00 & 15.00 & 18.30 & 17.00 & 17.40 & 15.30 \\
        \textbf{T} & 15.60 & 15.50 & 16.70 & 11.00 & 15.50 & 13.00 & 13.30 & 14.00 & 13.30 & 17.20 \\
        \textbf{t} & 20.00 & 17.30 & 19.10 & 17.90 & 15.00 & 14.00 & 10.00 & 16.00 & 17.10 & 18.10 \\
        \textbf{U} & 17.30 & 19.50 & 19.00 & 20.00 & 17.50 & 19.30 & 16.00 & 11.00 & 15.70 & 15.00 \\
        \textbf{u} & 10.70 & 16.00 & 20.00 & 18.00 & 17.60 & 10.00 & 17.10 & 10.00 & 16.90 & 13.00 \\
        \hline
    \end{tabular}
 
\end{table}
\end{document}